\documentclass[a4paper,12pt]{article}

\usepackage{color}
\usepackage{amssymb, amsmath}
\usepackage{hyperref}
\usepackage{tikz}
\usepackage{cite}

\newcommand\EFFACE[1]{}

\newtheorem{theorem}{Theorem}
\newtheorem{proposition}[theorem]{Proposition}
\newtheorem{definition}[theorem]{Definition}
\newtheorem{notation}[theorem]{Notation}
\newtheorem{lemma}[theorem]{Lemma}
\newtheorem{corollary}[theorem]{Corollary}

\newtheorem{claim}[theorem]{Claim}
\newtheorem{observation}[theorem]{Observation}

\newenvironment{proof}{
\par
\noindent {\bf Proof.}\rm}{\mbox{}\hfill$\square$\par\vskip 3mm}

\newcommand\SSS{\mathcal{S}}   
\newcommand\GGG{\mathcal{B}}   
\newcommand\FFFA{\mathcal{A}}   
   
\newcommand\XX{\mathcal{X}}   
\newcommand\XXX{\mathcal{C}}   
\newcommand\YY{\mathcal{Y}}   
   
\newcommand\YYYC{\mathcal{D}}   
\newcommand\ZZ{\mathcal{Z}}   
\newcommand\ZZZ{\mathcal{Z}}   
\newcommand\TTT{\mathcal{T}}   
\newcommand\EEE{\mathcal{E}}   
\newcommand\BBB{\mathcal{H}}   

\def\tf{\tilde{f}}
\def\tfs{\tf^*}

\def\NNNNN{\mathbb{N}}

\def\diam{{\rm diam}}
\def\rad{{\rm rad}}
\def\cost{{\rm cost}}

\makeatletter
\let\@fnsymbol\@arabic
\makeatother

\begin{document}


\newcommand\CTtriangle[2]{
\node[scale=0.4,draw,circle,fill=black] (v0) at (#1,#2){};
\draw[thick] (v0) to (#1-0.5,#2-1.5);
\draw[thick] (v0) to (#1+0.5,#2-1.5);
\draw[thick] (#1+0.5,#2-1.5) to (#1-0.5,#2-1.5);
}

\newcommand\CTtriangleOPT[2]{
\node[scale=0.4,draw,circle,fill=black] (v0) at (#1,#2){};
\draw[thick,dashed] (v0) to (#1-0.5,#2-1.5);
\draw[thick,dashed] (v0) to (#1+0.5,#2-1.5);
\draw[thick,dashed] (#1+0.5,#2-1.5) to (#1-0.5,#2-1.5);
}

\newcommand\CTzero[2]{
\node[scale=0.4,draw,circle,fill=black] (v0) at (#1,#2){};
}

\newcommand\CTun[2]{
\node[scale=0.4,draw,circle,fill=black] (v0) at (#1,#2){};
\node[scale=0.4,draw,circle,fill=black] (v1) at (#1,#2-1){};
\draw[thick] (v0) to (v1);
}

\newcommand\CTzeroplus[2]{
\node[scale=0.4,draw,circle,fill=black] (v0) at (#1,#2){};
\node[scale=0.4,draw,circle,fill=black] (v1) at (#1,#2-1){};
\draw[thick,dotted] (v0) to (v1);
}

\newcommand\CTunplus[2]{
\node[scale=0.4,draw,circle,fill=black] (v0) at (#1,#2){};
\node[scale=0.4,draw,circle,fill=black] (v1) at (#1-0.2,#2-1){};
\draw[thick] (v0) to (v1);
\draw[thick,dotted] (#1+0.1,#2-1) -- ++(0.4,0);
}

\newcommand\CTdeux[2]{
\node[scale=0.4,draw,circle,fill=black] (v0) at (#1,#2){};
\node[scale=0.4,draw,circle,fill=black] (v1) at (#1-0.3,#2-1){};
\node[scale=0.4,draw,circle,fill=black] (v2) at (#1+0.3,#2-1){};
\draw[thick] (v0) to (v1);
\draw[thick] (v0) to (v2);
}

\newcommand\CTdeuxmoins[2]{
\node[scale=0.4,draw,circle,fill=black] (v0) at (#1,#2){};
\node[scale=0.4,draw,circle,fill=black] (v1) at (#1-0.3,#2-1){};
\node[scale=0.4,draw,circle,fill=black] (v2) at (#1+0.3,#2-1){};
\draw[thick] (v0) to (v1);
\draw[thick,dashed] (v0) to (v2);
}

\newcommand\CTdeuxplus[2]{
\node[scale=0.4,draw,circle,fill=black] (v0) at (#1,#2){};
\node[scale=0.4,draw,circle,fill=black] (v1) at (#1-0.5,#2-1){};
\node[scale=0.4,draw,circle,fill=black] (v3) at (#1+0.5,#2-1){};
\draw[thick] (v0) to (v1);
\draw[thick] (v0) to (v3);
\draw[thick,dotted] (#1-0.15,#2-1) -- ++(0.4,0);
}

\newcommand\CTtrois[2]{
\node[scale=0.4,draw,circle,fill=black] (v0) at (#1,#2){};
\node[scale=0.4,draw,circle,fill=black] (v1) at (#1-0.6,#2-1){};
\node[scale=0.4,draw,circle,fill=black] (v2) at (#1,#2-1){};
\node[scale=0.4,draw,circle,fill=black] (v3) at (#1+0.6,#2-1){};
\draw[thick] (v0) to (v1);
\draw[thick] (v0) to (v2);
\draw[thick] (v0) to (v3);
}

\newcommand\CTunatrois[2]{
\node[scale=0.4,draw,circle,fill=black] (v0) at (#1,#2){};
\node[scale=0.4,draw,circle,fill=black] (v1) at (#1-0.6,#2-1){};
\node[scale=0.4,draw,circle,fill=black] (v2) at (#1,#2-1){};
\node[scale=0.4,draw,circle,fill=black] (v3) at (#1+0.6,#2-1){};
\draw[thick] (v0) to (v1);
\draw[thick,dashed] (v0) to (v2);
\draw[thick,dashed] (v0) to (v3);
}

\newcommand\CTtroisplus[2]{
\node[scale=0.4,draw,circle,fill=black] (v0) at (#1,#2){};
\node[scale=0.4,draw,circle,fill=black] (v1) at (#1-0.8,#2-1){};
\node[scale=0.4,draw,circle,fill=black] (v2) at (#1-0.2,#2-1){};
\node[scale=0.4,draw,circle,fill=black] (v3) at (#1+0.8,#2-1){};
\draw[thick] (v0) to (v1);
\draw[thick] (v0) to (v2);
\draw[thick] (v0) to (v3);
\draw[thick,dotted] (#1+0.15,#2-1) -- ++(0.4,0);
}

\newcommand\CTtroisplusD[2]{
\node[scale=0.4,draw,circle,fill=black] (v0) at (#1,#2){};
\node[scale=0.4,draw,circle,fill=black] (v1) at (#1-0.8,#2-1){};
\node[scale=0.4,draw,circle,fill=black] (v2) at (#1-0.2,#2-1){};
\node[scale=0.4,draw,circle,fill=black] (v3) at (#1+0.8,#2-1){};
\draw[thick,dashed] (v0) to (v1);
\draw[thick,dashed] (v0) to (v2);
\draw[thick,dashed] (v0) to (v3);
\draw[thick,dotted] (#1+0.15,#2-1) -- ++(0.4,0);
}

\newcommand\CTquatreplus[2]{
\node[scale=0.4,draw,circle,fill=black] (v0) at (#1,#2){};
\node[scale=0.4,draw,circle,fill=black] (v1) at (#1-1.1,#2-1){};
\node[scale=0.4,draw,circle,fill=black] (v2) at (#1-0.5,#2-1){};
\node[scale=0.4,draw,circle,fill=black] (v3) at (#1+0.1,#2-1){};
\node[scale=0.4,draw,circle,fill=black] (v4) at (#1+1.1,#2-1){};
\draw[thick] (v0) to (v1);
\draw[thick] (v0) to (v2);
\draw[thick] (v0) to (v3);
\draw[thick] (v0) to (v4);
\draw[thick,dotted] (#1+0.45,#2-1) -- ++(0.4,0);
}

\newcommand\CTtriangleL[3]{
\node[scale=0.4,draw,circle,fill=black] (v0) at (#1,#2){};
\draw[thick] (v0) to (#1-0.5,#2-1.5);
\draw[thick] (v0) to (#1+0.5,#2-1.5);
\draw[thick] (#1+0.5,#2-1.5) to (#1-0.5,#2-1.5);
\node[above] at (#1,#2-1.5) {#3};
}

\newcommand\CTgrandtriangleL[3]{
\node[scale=0.4,draw,circle,fill=black] (v0) at (#1,#2){};
\draw[thick] (v0) to (#1-0.8,#2-1.5);
\draw[thick] (v0) to (#1+0.8,#2-1.5);
\draw[thick] (#1+0.8,#2-1.5) to (#1-0.8,#2-1.5);
\node[above] at (#1,#2-1.5) {#3};
}

\newcommand\CTzeroL[3]{
\node[scale=0.4,draw,circle,fill=black] (v0) at (#1,#2){};
\node[above] at (#1,#2) {#3};
}

\newcommand\CTzeroLL[3]{
\node[scale=0.4,draw,circle,fill=black] (v0) at (#1,#2){};
\node[below] at (#1,#2) {#3};
}

\newcommand\CTunL[3]{
\node[scale=0.4,draw,circle,fill=black] (v0) at (#1,#2){};
\node[scale=0.4,draw,circle,fill=black] (v1) at (#1,#2-1){};
\draw[thick] (v0) to (v1);
\node[below] at (#1,#2-1) {#3};
}

\newcommand\CTunplusL[3]{
\node[scale=0.4,draw,circle,fill=black] (v0) at (#1,#2){};
\node[scale=0.4,draw,circle,fill=black] (v1) at (#1-0.2,#2-1){};
\node[above] at (#1-0.2,#2-3) {#3};
\draw[thick] (v0) to (v1);
\draw[thick,dotted] (#1+0.1,#2-1) -- ++(0.4,0);
}

\newcommand\CTdeuxL[4]{
\node[scale=0.4,draw,circle,fill=black] (v0) at (#1,#2){};
\node[scale=0.4,draw,circle,fill=black] (v1) at (#1-0.3,#2-1){};
\node[scale=0.4,draw,circle,fill=black] (v2) at (#1+0.3,#2-1){};
\draw[thick] (v0) to (v1);
\draw[thick] (v0) to (v2);
\node[below] at (#1-0.3,#2-1) {#3};
\node[below] at (#1+0.3,#2-1) {#4};
}

\newcommand\CTdeuxmoinsL[4]{
\node[scale=0.4,draw,circle,fill=black] (v0) at (#1,#2){};
\node[scale=0.4,draw,circle,fill=black] (v1) at (#1-0.3,#2-1){};
\node[scale=0.4,draw,circle,fill=black] (v2) at (#1+0.3,#2-1){};
\draw[thick] (v0) to (v1);
\draw[thick,dashed] (v0) to (v2);
\node[below] at (#1-0.3,#2-1) {#3};
\node[below] at (#1+0.3,#2-1) {#4};
}

\newcommand\CTdeuxplusL[4]{
\node[scale=0.4,draw,circle,fill=black] (v0) at (#1,#2){};
\node[scale=0.4,draw,circle,fill=black] (v1) at (#1-0.5,#2-1){};
\node[scale=0.4,draw,circle,fill=black] (v3) at (#1+0.5,#2-1){};
\draw[thick] (v0) to (v1);
\draw[thick] (v0) to (v3);
\node[above] at (#1-0.5,#2-3) {#3};
\node[above] at (#1+0.5,#2-3) {#4};
\draw[thick,dotted] (#1-0.15,#2-1) -- ++(0.4,0);
}

\newcommand\CTtroisL[5]{
\node[scale=0.4,draw,circle,fill=black] (v0) at (#1,#2){};
\node[scale=0.4,draw,circle,fill=black] (v1) at (#1-0.6,#2-1){};
\node[scale=0.4,draw,circle,fill=black] (v2) at (#1,#2-1){};
\node[scale=0.4,draw,circle,fill=black] (v3) at (#1+0.6,#2-1){};
\draw[thick] (v0) to (v1);
\draw[thick] (v0) to (v2);
\draw[thick] (v0) to (v3);
\node[below] at (#1-0.6,#2-1) {#3};
\node[below] at (#1,#2-1) {#4};
\node[below] at (#1+0.6,#2-1) {#5};
}

\newcommand\CTunatroisL[5]{
\node[scale=0.4,draw,circle,fill=black] (v0) at (#1,#2){};
\node[scale=0.4,draw,circle,fill=black] (v1) at (#1-0.6,#2-1){};
\node[scale=0.4,draw,circle,fill=black] (v2) at (#1,#2-1){};
\node[scale=0.4,draw,circle,fill=black] (v3) at (#1+0.6,#2-1){};
\draw[thick] (v0) to (v1);
\draw[thick,dashed] (v0) to (v2);
\draw[thick,dashed] (v0) to (v3);
\node[below] at (#1-0.6,#2-1) {#3};
\node[below] at (#1,#2-1) {#4};
\node[below] at (#1+0.6,#2-1) {#5};
}

\newcommand\CTtroisplusL[5]{
\node[scale=0.4,draw,circle,fill=black] (v0) at (#1,#2){};
\node[scale=0.4,draw,circle,fill=black] (v1) at (#1-0.8,#2-1){};
\node[scale=0.4,draw,circle,fill=black] (v2) at (#1-0.2,#2-1){};
\node[scale=0.4,draw,circle,fill=black] (v3) at (#1+0.8,#2-1){};
\draw[thick] (v0) to (v1);
\draw[thick] (v0) to (v2);
\draw[thick] (v0) to (v3);
\node[below] at (#1-0.8,#2-1) {#3};
\node[below] at (#1-0.2,#2-1) {#4};
\node[below] at (#1+0.8,#2-1) {#5};
\draw[thick,dotted] (#1+0.15,#2-1) -- ++(0.4,0);
}

\newcommand\CTtroisplusDL[5]{
\node[scale=0.4,draw,circle,fill=black] (v0) at (#1,#2){};
\node[scale=0.4,draw,circle,fill=black] (v1) at (#1-0.8,#2-1){};
\node[scale=0.4,draw,circle,fill=black] (v2) at (#1-0.2,#2-1){};
\node[scale=0.4,draw,circle,fill=black] (v3) at (#1+0.8,#2-1){};
\draw[thick,dashed] (v0) to (v1);
\draw[thick,dashed] (v0) to (v2);
\draw[thick,dashed] (v0) to (v3);
\draw[thick,dotted] (#1+0.15,#2-1) -- ++(0.4,0);
\node[below] at (#1-0.8,#2-1) {#3};
\node[below] at (#1-0.2,#2-1) {#4};
\node[below] at (#1+0.8,#2-1) {#5};
}

\newcommand\CTquatreL[6]{
\node[scale=0.4,draw,circle,fill=black] (v0) at (#1,#2){};
\node[scale=0.4,draw,circle,fill=black] (v1) at (#1-0.3,#2-1){};
\node[scale=0.4,draw,circle,fill=black] (v2) at (#1-0.9,#2-1){};
\node[scale=0.4,draw,circle,fill=black] (v3) at (#1+0.3,#2-1){};
\node[scale=0.4,draw,circle,fill=black] (v4) at (#1+0.9,#2-1){};
\draw[thick] (v0) to (v1);
\draw[thick] (v0) to (v2);
\draw[thick] (v0) to (v3);
\draw[thick] (v0) to (v4);
\node[below] at (#1-0.9,#2-1) {#3};
\node[below] at (#1-0.3,#2-1) {#4};
\node[below] at (#1+0.3,#2-1) {#5};
\node[below] at (#1+0.9,#2-1) {#5};
}

\newcommand\CTquatreplusL[6]{
\node[scale=0.4,draw,circle,fill=black] (v0) at (#1,#2){};
\node[scale=0.4,draw,circle,fill=black] (v1) at (#1-1.1,#2-1){};
\node[scale=0.4,draw,circle,fill=black] (v2) at (#1-0.5,#2-1){};
\node[scale=0.4,draw,circle,fill=black] (v3) at (#1+0.1,#2-1){};
\node[scale=0.4,draw,circle,fill=black] (v4) at (#1+1.1,#2-1){};
\draw[thick] (v0) to (v1);
\draw[thick] (v0) to (v2);
\draw[thick] (v0) to (v3);
\draw[thick] (v0) to (v4);
\draw[thick,dotted] (#1+0.45,#2-1) -- ++(0.4,0);
\node[below] at (#1-1.1,#2-1) {#3};
\node[below] at (#1-0.5,#2-1) {#4};
\node[below] at (#1+0.1,#2-1) {#5};
\node[below] at (#1+1.1,#2-1) {#5};
}



\title{On the Broadcast Independence Number\\ of Locally Uniform 2-Lobsters}

\author{Messaouda AHMANE~\thanks{Faculty of Mathematics, Laboratory L'IFORCE, University of Sciences and Technology
Houari Boumediene (USTHB), B.P.~32 El-Alia, Bab-Ezzouar, 16111 Algiers, Algeria.}
\and Isma BOUCHEMAKH~\footnotemark[1]
\and \'Eric SOPENA~\thanks{Univ. Bordeaux, Bordeaux INP, CNRS, LaBRI, UMR5800, F-33400 Talence, France.}
}

\maketitle

\abstract{
Let~$G$ be a simple undirected graph.
A broadcast on~$G$ is
a function $f : V(G)\rightarrow\NNNNN$ such that $f(v)\le e_G(v)$ holds for every vertex~$v$ of~$G$, 
where $e_G(v)$ denotes the eccentricity of~$v$ in~$G$, that is, the maximum distance from~$v$ to any other vertex of~$G$.
The cost of~$f$ is the value $\cost(f)=\sum_{v\in V(G)}f(v)$.
A broadcast~$f$ on~$G$ is independent if for every two distinct vertices $u$ and~$v$ in~$G$, $d_G(u,v)>\max\{f(u),f(v)\}$,
where $d_G(u,v)$ denotes the distance between $u$ and~$v$ in~$G$.
The broadcast independence number of~$G$ is then defined as the maximum cost of an independent broadcast on~$G$.

A caterpillar is a tree such that, after the removal of all leaf vertices, the remaining graph is a non-empty path.
A lobster is a tree such that, after the removal of all leaf vertices, the remaining graph is a caterpillar.
In [M. Ahmane, I. Bouchemakh and E. Sopena.
On the Broadcast Independence Number of Caterpillars.
Discrete Applied Mathematics, in press (2018)], we studied independent broadcasts of caterpillars.
In this paper, carrying on with this line of research, we consider independent broadcasts of lobsters 
and give an explicit formula for the
broadcast independence number of a family of lobsters called locally uniform $2$-lobsters.
}

\medskip



\section{Introduction}
\label{sec:intro}

All the graphs we consider in this paper are simple and loopless undirected graphs. 
We denote by $V(G)$ and $E(G)$ the set of vertices and the set of edges of a graph~$G$, respectively.

For any two vertices $u$ and $v$ of $G$,  the \emph{distance} $d_G(u,v)$ between $u$ and $v$ in $G$
is the length (number of edges) of a shortest path joining $u$ and $v$.
The \emph{eccentricity} $e_G(v)$ of a vertex $v$ in $G$ is
the maximum distance from $v$ to any other vertex of $G$. 
The minimum eccentricity in $G$ is the \emph{radius} $\rad(G)$ of $G$, while the maximum eccentricity in $G$ is the
\emph{diameter} $\diam(G)$ of $G$. 

A function $f : V(G)\rightarrow\{0,\dots,\diam(G)\}$ is a \emph{broadcast on $G$} 
if for every vertex $v$ of $G$, $f(v)\le e_G(v)$.
The value $f(v)$ is  called the \emph{$f$-value of $v$}.
Given a broadcast $f$ on $G$, an \emph{$f$-broadcast vertex} is a vertex $v$ with $f(v)>0$.
The set of all $f$-broadcast vertices is denoted $V_f^+$.
If $u\in V_f^+$ is a broadcast vertex, $v\in V(G)$ and $d_G(u,v)\le f(u)$, we say that \emph{$u$ $f$-dominates $v$}.
In particular, every $f$-broadcast vertex $f$-dominates itself.
The \emph{cost} $\cost(f)$ of a broadcast $f$ on $G$ is given by
$$\cost(f)=\sum_{v\in V(G)}f(v)=\sum_{v\in V_f^+}f(v).$$

A broadcast $f$ on $G$ is a \emph{dominating broadcast} if every vertex of $G$
is $f$-dominated by some vertex of $V_f^+$.
The minimum cost of a dominating broadcast on $G$ is the \emph{broadcast domination number} of $G$,
denoted $\gamma_b(G)$.
A 
broadcast $f$ on $G$ is an \emph{independent broadcast} if every $f$-broadcast vertex 
is $f$-dominated only by itself.
The maximum cost of an independent broadcast on $G$ is the \emph{broadcast independence number} of $G$,
denoted $\beta_b(G)$.
An independent broadcast on $G$ with cost $\beta$ is an independent \emph{$\beta$-broadcast}.
An independent $\beta_b(G)$-broadcast on $G$ is an \emph{optimal} independent broadcast.
Note here that any optimal independent broadcast is necessarily a dominating broadcast.

The notions of broadcast domination and broadcast independence were introduced by D.J.~Erwin
in his Ph.D. thesis~\cite{E01} under the name of \emph{cost domination} and \emph{cost independence}, respectively.
During the last decade,
broadcast domination has been investigated by several authors (see e.g.~\cite{BB,BBF18,BHHM04,BF17,BMGY,BMT13,BS11,BS09,CHM11,DDH09,DEHHH06,E04,GM17,HM14,HL06,HM09,LM15,MR17,MT,MW13,MW15,S08,SK14}),
while independent broadcast domination has attracted much less 
attention (see~\cite{ABS18,BZ14}), until the recent work of Bessy and Rautenbach.
In~\cite{BR18b}, these authors prove that $\beta_b(G)\le 4\alpha(G)$ for every graph $G$,
where $\alpha(G)$ denotes the independence number of $G$, that is, 
the maximum cardinality of an independent set in $G$.
In~\cite{BR18b}, they prove that $\beta_b(G) < 2\alpha(G)$ whenever $G$ has 
girth at least~6 and minimum degree at least~3, or
girth at least~4 and minimum degree at least~5.
Answering questions posed in \cite{H06} and~\cite{DEHHH06},
they prove in~\cite{BR18a} that deciding whether $\beta_b(G)\ge k$ 
for a given planar graph with maximum degree four and a given positive integer $k$ 
is an NP-complete problem,
and, using an approach based on dynamic programming,
they prove that determining the value of $\beta_b(T)$ for a tree $T$ of order $n$
can be done in time $O(n^9)$.

\medskip

Our goal, initiated in~\cite{ABS18}, is to give explicit formulas for $\beta_b(T)$,
whenever $T$ belongs to some particular subclass of trees, that can be computed
in (hopefully) linear time.
Recall that a \emph{caterpillar} is a tree such that deleting all its pendent
vertices leaves a simple path, called the \emph{spine} of the caterpillar. 
A \emph{lobster} is then a tree such that deleting all its pendent
vertices leaves a caterpillar. 
The spine of such a lobster is the spine of the
so-obtained caterpillar. 
A vertex belonging to the spine of a caterpillar, or of a lobster,
is called a \emph{spine-vertex} and an \emph{internal} spine-vertex
if it is not an end vertex of the spine.
The \emph{length} of a lobster $L$ is the length (number of edges) of its spine.

Note that if $L$ is a lobster of length~0, then the unique spine-vertex of $L$
must be of degree at least~2, since otherwise, deleting all leaves of $L$ would
leave a single edge, which is not a caterpillar.
Hence, $\diam(L)=k+4$ for every lobster $L$ of length~$k$.

In~\cite{ABS18}, we gave an explicit formula for the broadcast independence number
of caterpillars having no two consecutive internal spine vertices of degree~2.
The aim of this paper is to pursue the study of independent broadcasts of trees by
considering the case of locally uniform $2$-lobsters.

Let $G$ be a graph  and $A\subset V(G)$, $|A|\ge 2$, be a set of pairwise antipodal vertices in $G$,
that is, at distance $\diam(G)$ from each other.
The function $f$ defined by $f(u)=\diam(G)-1$ for every vertex $u\in A$,
and $f(v)=0$ for every vertex $v\not\in A$,
is clearly an independent $|A|(\diam(G)-1)$-broadcast on $G$. 

\begin{observation}[Dunbar {\it et al.}\cite{DEHHH06}]
For every graph $G$ of order at least 2
and every set $A\subset V(G)$, $|A|\ge 2$, of pairwise antipodal vertices in $G$, 
$\beta_b(G)\ge |A|(\diam(G)-1)\ge 2(\diam(G)-1)$.
\label{obs:2(d-1)}
\end{observation}

\medskip

In this paper, we determine the broadcast independence
number of locally uniform $2$-lobsters.
The paper is organised as follows. We introduce in the next section the
main definitions and a few preliminary results.
We then consider in Section~\ref{sec:2-lobsters} the case of locally uniform
$2$-lobsters and prove our main result, which gives an explicit
formula for the broadcast independence number of such lobsters.
We then propose some concluding remarks in Section~\ref{sec:discussion}.


\section{Preliminaries}
\label{sec:preliminaries}

Let $G$ be a graph and $H$ be a subgraph of $G$.
Since $d_H(u,v)\ge d_G(u,v)$ for every two vertices $u,v\in V(H)$,
every independent broadcast $f$ on $G$ satisfying $f(u)\le e_H(u)$ for every
vertex $u\in V(H)$ is an independent broadcast on $H$.
Hence we have:

\begin{observation}
If $H$ is a subgraph of $G$ and $f$ is an independent broadcast on $G$ 
satisfying $f(u)\le e_H(u)$ for every vertex $u\in V(H)$, then the restriction
$f_H$ of $f$ to $V(H)$ is an independent broadcast on $H$.
\label{obs:subgraph}
\end{observation}

For any independent broadcast $f$ on a graph $G$, and any subgraph $H$ of $G$,
we denote by $f^*(H)$ the \emph{$f$-value} of $H$ defined as
$$f^*(H)=\sum_{v\in V(H)}f(v).$$
Observe that $f^*(G)=\cost(f)$.

\medskip

The following lemma shows that, for any graph $G$ of order at least~3,
if $v$ is a vertex of $G$ having at least one pendent neighbour,
then no independent broadcast $f$ on $G$ with $f(v)>0$ can be optimal.

\begin{lemma}
Let $G$ be a graph of order at least~3 and $v$ be a vertex of $G$ having a pendent neighbour~$u$.
If $f$ is an
independent broadcast on $G$ with $f(v)>0$, then
there exists an independent broadcast $f'$ on $G$ with $\cost(f')>\cost(f)$.
\label{lem:branch-0}
\end{lemma}

\begin{proof}
The mapping $f'$
defined by $f'(u)=f(v)+1$, $f'(v)=0$ and $f'(w)=f(w)$ for every vertex $w\in V(G)\setminus\{u,v\}$
is clearly an independent broadcast on $G$ with $\cost(f')>\cost(f)$.
\end{proof}

The following lemma was given in~\cite{ABS18}.
However, we include its proof here for the sake of completeness.

\begin{lemma}
Let $T$ be a tree of order at least~3, and $T'$ be a subtree of $T$ of order at least 2, with root~$r$.
Let $f$ be an optimal independent broadcast on $T$.
If $r$ is an $f$-broadcast vertex, then $T'$ contains at least one other $f$-broadcast vertex.
In particular, this implies that if $T'$ is a subtree of height~$h$, that is, $e_{T'}(r)=h$, then $f(r)<h$.
\label{lem:end-leaves-tree}
\end{lemma}

\begin{proof}
Suppose to the contrary that $f(r)>0$ and $f(u)=0$ for every vertex $u\in V(T')\setminus\{r\}$.
Let $t'=e_{T'}(r)$ and $\overline{t'}=e_{T-(T'-r)}(r)$. 

If $f(r)<t'$, the independent broadcast $f'$ given by $f'(v)=f(r)$ for some vertex $v$ in $T'$ with $d_{T'}(r,v)=t'$
and $f'(u)=f(u)$ for every vertex $u\in V'(T)\setminus\{v\}$ is such that $\cost(f')=\cost(f)+f(r)$, 
contradicting the optimality of $f$.

If $f(r)\ge\overline{t'}$, then $r$ is the unique $f$-broadcast vertex, which implies
$\cost(f)<2(\diam(T)-1)$, again contradicting the optimality of $f$ by Observation~\ref{obs:2(d-1)}.

Hence $\overline{t'}>f(r)\ge t'$.
Let now $v$ be any neighbour of $r$ in $T'$.
Since $\overline{t'}>f(r)\ge t'$, we have $e_T(v)=e_T(r)+1=\overline{t'}+1>f(r)+1$.
The function $f'$ defined by $f'(r)=0$, $f'(v)=f(r)+1$
and $f'(u)=f(u)$ for every vertex $u\in V(T)\setminus\{r,v\}$ is therefore
an independent broadcast on $T$ with $\cost(f')=\cost(f)+1$, contradicting the optimality of $f$.

This completes the proof.
\end{proof}

\medskip

In order to formally define locally uniform lobsters,
and then locally uniform $2$-lobsters,
we introduce some notation.

\begin{notation}[$\SSS_1,\ \SSS_2$]
{\rm A tree $T$ rooted at a vertex $r$ is of type $\SSS_1$ if every leaf of $T$
is at distance $1$ from $r$, which means that $T$ is a star with center $r$.
A tree $T$ rooted at a vertex $r$ is of type $\SSS_2$ if every leaf of $T$
is at distance $2$ from $r$.
}\end{notation}

Let $L$ be a lobster with spine $v_0\dots v_k$, $k\ge 0$.
The \emph{subtree} of $v_i$, $0\le i\le k$, denoted $S_i$, is the maximal subtree
of $L$ rooted at $v_i$ that contains no spine-vertex except $v_i$.
A \emph{spine-subtree} of $L$ is a subtree of some $v_i$, $0\le i\le k$.
A \emph{branch} of a spine-subtree $S_i$ is a maximal subtree of $S_i$ containing $v_i$
and exactly one neighbour of $v_i$. Therefore, if $v_i$ has degree $d$ in $S_i$, then
$S_i$ has $d$ distinct branches.

A locally uniform lobster is then defined as follows.

\begin{definition}[Locally uniform lobster]
{\rm A lobster $L$ is \emph{locally uniform} if every spine-subtree of $L$
is of type either $\SSS_1$ or $\SSS_2$. In other words, all branches of any
spine-vertex have the same depth.
}\end{definition}

The following observation directly follows from this definition.

\begin{observation}
If $L$ is a locally uniform lobster with spine $v_0\dots v_k$, $k\ge 0$,
then both spine-subtrees $S_0$ and $S_k$ are of type $\SSS_2$.
\label{obs:ends}
\end{observation}

Indeed, if $S_0$ or $S_k$ is of type $\SSS_1$, then $v_0$ or $v_k$
is a leaf of the caterpillar obtained by deleting all leaves of $L$,
which implies that $v_0\dots v_k$ is not the spine of $L$, a contradiction.

Observe that Lemma~\ref{lem:end-leaves-tree} implies in particular the following result
for locally uniform lobsters.

\begin{corollary}
If $L$ is a locally uniform lobster with spine $v_0\dots v_k$, $k\ge 0$,
and $f$ is an optimal independent broadcast on $L$, then
the two following conditions hold.
\begin{enumerate}
\item If $v$ is a vertex having a pendent neighbour, then $f(v)=0$.
\item For every $i$, $0\le i\le k$, $f(v_i)=0$ if $S_i$ is of type $\SSS_1$, and $f(v_i)\le 1$ if 
$S_i$ is of type $\SSS_2$.
\end{enumerate}
\label{cor:lobster}
\end{corollary}

Moreover, the following lemma says that for every 
optimal independent broadcast on 
a locally uniform lobster with spine $v_0\dots v_k$, $k\ge 0$,
both the spine-subtrees $S_0$ and $S_k$ contain an $f$-broadcast vertex.

\begin{lemma}
If $L$ is a locally uniform lobster with spine $v_0\dots v_k$, $k\ge 0$,
and $f$ is an optimal independent broadcast on $L$, then
$f^*(S_0)>0$ and $f^*(S_k)>0$.
\label{lem:broadcast-vertex-in-S0-Sk}
\end{lemma}

\begin{proof}
It is enough to prove the result for $S_0$.
Assume to the contrary that $f^*(S_0)=0$, and let $v$ be a vertex of $L$ that
$f$-dominates the leaves of $S_0$.
Since $f^*(S_0)=0$, we necessarily have $f(v)\geq 4$
which implies that $v$ is unique.
By Corollary~\ref{cor:lobster}, $v$ must be a leaf of $L$.
Let $\ell$ be any leaf of $S_0$.

Let $S$ denote the spine-subtree containing $v$.
If $S$ is of type $\SSS_1$ and $f(v)+d_L(\ell,v)>\diam(L)+1$, 
or $S$ is of type $\SSS_2$ and $f(v)+d_L(\ell,v)>\diam(L)+3$, 
then $v$ is the unique $f$-broadcast
vertex of $L$, which contradicts the optimality of $f$ by Observation~\ref{obs:2(d-1)}.
We now define the mapping $f'$ on $V(L)$ given by
$f'(u)=f(u)$ for every vertex $u\notin\{v,\ell\}$,
$f'(v)=0$, and
$f'(\ell)=f(v)+d_L(\ell,v)-2$.
The mapping $f'$ is clearly an independent broadcast on $L$ and,
since $d_L(v,\ell)\geq 4$, we get $\cost(f') > \cost(f)$,
contradicting the optimality of~$f$.
\end{proof}


\begin{figure}
\begin{center}
\begin{tikzpicture}
\draw[thick] (0.1,0) to (11,0);
\CTzero{0.1}{0} \CTun{-0.3}{-1}  \CTdeux{0.5}{-1}
\draw[thick] (0.1,0) to (-0.3,-1);
\draw[thick] (0.1,0) to (0.5,-1);
\CTtrois{2}{0}
\CTun{3.5}{-1}  \CTun{4.1}{-1} 
\CTun{4.7}{-1}  \CTdeux{5.6}{-1}   \CTtrois{7.1}{-1}
\CTzero{5.3}{0}
\draw[thick] (5.3,0) to (3.5,-1);
\draw[thick] (5.3,0) to (4.1,-1);
\draw[thick] (5.3,0) to (4.7,-1);
\draw[thick] (5.3,0) to (5.6,-1);
\draw[thick] (5.3,0) to (7.1,-1);
\CTdeux{8.6}{0}
\CTdeux{9.8}{0}
\CTdeux{11}{0}  \CTun{10.7}{-1}  \CTun{11.3}{-1}
\end{tikzpicture}
\caption{\label{fig:sample-2-lobster}A sample locally uniform $2$-lobster.}
\end{center}

\end{figure}
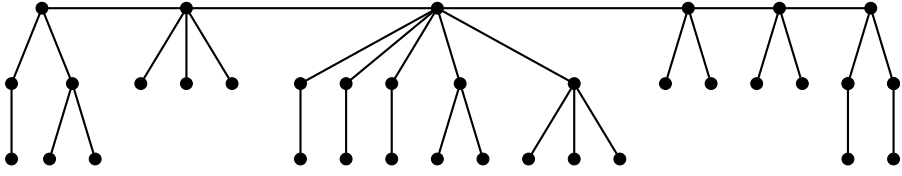

We now define $2$-lobsters and locally uniform $2$-lobsters.


\begin{definition}[$2$-lobster]
{\rm A lobster $L$ 
is a \emph{$2$-lobster} if every spine-subtree of $L$ has at least two branches.
}\end{definition}

\begin{definition}[Locally uniform $2$-lobster]
{\rm A \emph{locally uniform $2$-lobster} is a $2$-lobster which is locally uniform
(see Figure~\ref{fig:sample-2-lobster}).
}\end{definition}

Due to their special structure, we can
improve the lower bound on the broadcast independence number of 
locally uniform $2$-lobsters of length $k\ge 1$.

\begin{observation}
For every locally uniform $2$-lobster $L$ of length $k\ge 1$, 
$$\beta_b(L)\ge 2(k-1) + 12 = 2(\diam(L)-1) + 4.$$
\label{obs:lower-bound-2-lul}
\end{observation}

To see that, consider the function $f$ on $V(L)$ defined as follows.
For each branch of $S_0$ and $S_k$, pick one leaf and set its $f$-broadcast
value to~$3$, and, for each branch of every $S_i$, $1\le i\le k-1$, if $k>1$,
pick one leaf and set its $f$-broadcast value to~$1$. 
The mapping $f$ is clearly an independent broadcast on $L$ and,
since both $S_0$ and $S_k$ are of type $\SSS_2$ and every spine-subtree of $L$
has at least two branches, we get $\cost(f)=2(k-1) + 12$.

\section{Independent broadcasts of locally uniform $2$-lobsters}
\label{sec:2-lobsters}

In this section we determine the broadcast independence number
of locally uniform $2$-lobsters.
Recall that by Observation~\ref{obs:lower-bound-2-lul}, 
$\beta_b(L)>2(\diam(L)-1)$ for every locally uniform $2$-lobster $L$ of length $k\ge 1$
(the special case of locally uniform $2$-lobsters of length~0 will be considered
separately, in Lemma~\ref{lem:beta-star-zero}).

We first introduce some notation and define different types of spine-subtrees
in Subsection~\ref{sub:types}.
We then define the value $\beta^*(L)$ for every locally uniform $2$-lobster $L$
in Subsection~\ref{sub:beta-star}, prove that every such lobster $L$ admits an
independent $\beta^*(L)$-broadcast in Subsection~\ref{sub:lower-bound}
and that it cannot admit any independent broadcast with cost
strictly greater than $\beta^*(L)$ in Subsection~\ref{sub:upper-bound}.
This allows us to finally state our main result in Subsection~\ref{sub:main-result}.

%


\subsection{Different types of spine-subtrees}
\label{sub:types}

Let $L$ be a locally uniform $2$-lobster with spine $v_0\dots v_k$, $k\ge 0$.
Two spine-subtrees $S_i$ and $S_{i+1}$, $0\le i\le k-1$, are called \emph{neighbouring spine-subtrees}.
Moreover, we say that $S_i$ \emph{precedes} $S_{i+1}$, and that $S_{i+1}$ \emph{follows} $S_i$.
A \emph{sequence} of $p$ spine-subtrees, $p\ge 2$, is a sequence of consecutive spine-subtrees
of the form $S_i\dots S_{i+p-1}$ for some $i$, $0\le i\le k-p+1$.

We will say that two independent broadcasts $f_1$ and $f_2$ on a locally uniform $2$-lobster $L$ are \emph{similar}
if their values on each spine-subtree of $L$ are equal, that is,
$f_1^*(S_i)=f_2^*(S_i)$ for every $i$, $0\le i\le k$. 
Observe that any two similar independent broadcasts have the same cost.


A \emph{1-leaf} of $L$ is a pendent vertex of $L$ adjacent to a spine-vertex.
A pendent vertex which is not a 1-leaf is a \emph{2-leaf} 
(recall that every pendent vertex is at distance at most~2 from a spine-vertex).
An \emph{only-leaf} is a leaf whose neighbour has only one leaf neighbour.
Therefore, an only-leaf in a locally uniform $2$-lobster
is necessarily a 2-leaf, and is then called a \emph{$2$-only-leaf}.
Two leaves having the same neighbour are said to be \emph{sister-leaves}.

\begin{notation}[$\lambda_1$, $\lambda_2$, $\lambda_2^*$]{\rm 
For every $i$, $0\le i\le k$, we denote by $\lambda_1(S_i)$, $\lambda_2(S_i)$ and $\lambda_2^*(S_i)$, 
the number  of 1-leaves, of 2-leaves, and of 2-only-leaves of $S_i$, respectively.
Moreover, we extend these three functions to the whole lobster $L$, by letting
$$\lambda_1(L)=\sum_{i=0}^{i=k}\lambda_1(S_i),\ \ 
\lambda_2(L)=\sum_{i=0}^{i=k}\lambda_2(S_i),\ \
\mbox{and}\ \  \lambda_2^*(L)=\sum_{i=0}^{i=k}\lambda_2^*(S_i).$$
}\end{notation}

Let $v_i$ be a spine-vertex of $L$ with $t$ non-spine neighbours, denoted $w_i^1,\dots,w_i^{t}$.
For every $j$, $1\le j\le t$, the branch $B_i^j$ of $v_i$ is the maximal spine-subtree of $S_i$, rooted at $v_i$,
containing the edge $v_iw_i^j$ but no edge $v_iw_i^{j'}$ with $j'\neq j$.
We then define two types of branches. 

\begin{notation}[$\BBB_1$, $\BBB_2$, $\alpha_1$, $\alpha_2$, $\alpha_2^*$]
{\rm A branch is of type $\BBB_1$ if it does not contain any 2-leaf, and
of type $\BBB_2$ if it does not contain any 1-leaf.
For every spine-subtree $S_i$, $0\le i\le k$, we denote by $\alpha_1(S_i)$ and
$\alpha_2(S_i)$ the number of branches of $S_i$ of type $\BBB_1$ and of type $\BBB_2$, respectively.
Moreover, we denote by $\alpha_2^{*}(S_i)$ the number of branches of $S_i$ of type $\BBB_2$
having at most two 2-leaves.
}\end{notation}

Since all branches of any spine-subtree of a locally uniform $2$-lobster are of the same type,
we get $\alpha_1(S_i)\ge 2$, $\alpha_2(S_i)=\alpha_2^*(S_i)=0$, if $S_i$ is of type $\SSS_1$,
and $\alpha_1(S_i)=0$, $\alpha_2(S_i)\ge 2$, $\alpha_2^*(S_i)\ge 0$, if $S_i$ is of type $\SSS_2$.

\begin{notation}[$b_i$]
{\rm For every $i$, $0\le i\le k$, we denote by $b_i$ the number of branches of the spine-subtree~$S_i$.
}\end{notation}

Observe that $b_i=\deg_L(v_i)-2$ if $1\le i\le k-1$, and $b_i=\deg_L(v_i)-1$ if $i\in\{0,k\}$.

In order to define various types of spine-subtrees, we will use
the following notation.

\begin{notation}[Operators on types of spine-subtrees]
{\rm Let $\XX$, $\YY$ and $\ZZ$ be any types of spine-subtrees.
We then define the following types.
\begin{itemize}
\item $\overline{\XX}$.\\
A spine-subtree $S$
is of type $\overline{\XX}$ if $S$ is not of type $\XX$.
\item $\XX|\YY$.\\
A spine-subtree $S$
is of type $\XX|\YY$ if $S$ is of type $\XX$ or~$\YY$.
\item $\XX.\YY$, $\XX\YY$.\\
A sequence of two spine-subtrees $SS'$ is of type $\XX.\YY$, or simply $\XX\YY$,
if $S$ is of type $\XX$ and $S'$ is of type $\YY$.
\item $\XX[P_1,\dots,P_p]$.\\
For any properties $P_1,\dots,P_p$, $p\ge 1$, 
a spine-subtree $S$ is of type $\XX[P_1,\dots,P_p]$
if $S$ is a spine-subtree of type $\XX$ satisfying properties $P_1,\dots,P_p$.
For instance, a spine-subtree $S$ is of type $\SSS_2[\lambda_2\ge 5,\alpha_2^{*}\le 3]$ 
if $S$ is a spine-subtree of type $\SSS_2$ with at least five leaves,
having at most three branches with at most two leaves.
Similarly, a branch of type $\YY[P_1,\dots,P_p]$ is a branch of type $\YY$ satisfying properties $P_1,\dots,P_p$.
For instance, a branch of type $\BBB_2[\lambda_2=3]$ is a branch of type $\BBB_2$ having three 2-leaves.
\item $\langle\XX\rangle\YY$, $\YY\langle\ZZ\rangle$, $\langle\XX\rangle\YY\langle\ZZ\rangle$.\\
A spine-subtree $S$ is of type
$\langle\XX\rangle\YY$ (resp. $\YY\langle\ZZ\rangle$)
if $S$ is a spine-subtree of type $\YY$ and the spine-subtree $S'$ preceding $S$ (resp. following $S$)
is of type $\XX$ (resp. $\ZZ$).
A spine-subtree $S$ is then of type $\langle\XX\rangle\YY\langle\ZZ\rangle$ 
if $S$ is of type $\langle\XX\rangle\YY$ and of type $\YY\langle\ZZ\rangle$.
\item $\varnothing$.\\
Slightly abusing the notation, we use the symbol $\varnothing$ to denote an ``empty spine-subtree'', so that, for instance,
a spine-subtree $S$ is of type $\langle\varnothing\rangle\YY=$ (resp. $\YY\langle\varnothing\rangle$), if $S=S_0$ (resp. $S=S_k$)
and $S$ is of type~$\YY$.
\item $\{\XX_1\dots \XX_p\}^+$, $\{\XX_1\dots \XX_p\}^*$.\\
For any types of spine-subtrees $\XX_1,\dots,\XX_p$, $p\ge 1$,
a sequence of spine-subtrees $S_i,\dots,S_{i+pj}$, 
$0\le i\le k-pj$, $0\le j\le \lfloor\frac{k-i}{p}\rfloor$, is of type $\{\XX_1\dots \XX_p\}^+$,
if every spine-subtree $S_\ell$, $i\le \ell \le i+pj$ is of type $\XX_{\ell-i+1\pmod p}$,
and none of the sequences $S_{i-p},\dots,S_i,\dots,S_{i+pj}$ and $S_i,\dots,S_{i+pj},\dots,S_{i+pj+p}$
is of type $\{\XX_1\dots \XX_p\}^+$ (the sequence is thus maximal).
Moreover, we will denote by $\{\XX_1\dots \XX_p\}^*$ the type $\varnothing|\{\XX_1\dots \XX_p\}^+$.
\end{itemize}
}\end{notation}

\medskip

Our aim now is twofold. We will first construct, for any locally uniform $2$-lobster $L$,
an independent broadcast $f^*$ on $L$ with $\cost(f^*)=\beta^*(L)$,
for some value $\beta^*(L)$,
and then prove that the value $\beta^*(L)$ is
the optimal cost of an independent broadcast on $L$.

The independent broadcast $f^*$ will be constructed in four steps,
that is, we will construct a sequence of independent broadcasts
$f_1,\dots,f_4$, with $\cost(f_i)\le\cost(f_{i+1})$ for
every $i$, $1\le i\le 3$, and then set $f^*=f_4$.
Each step will consist in modifying the broadcast values of some
vertices, according to the type of the spine-subtree,
or of the sequence of spine-subtrees, they belong to.

\medskip

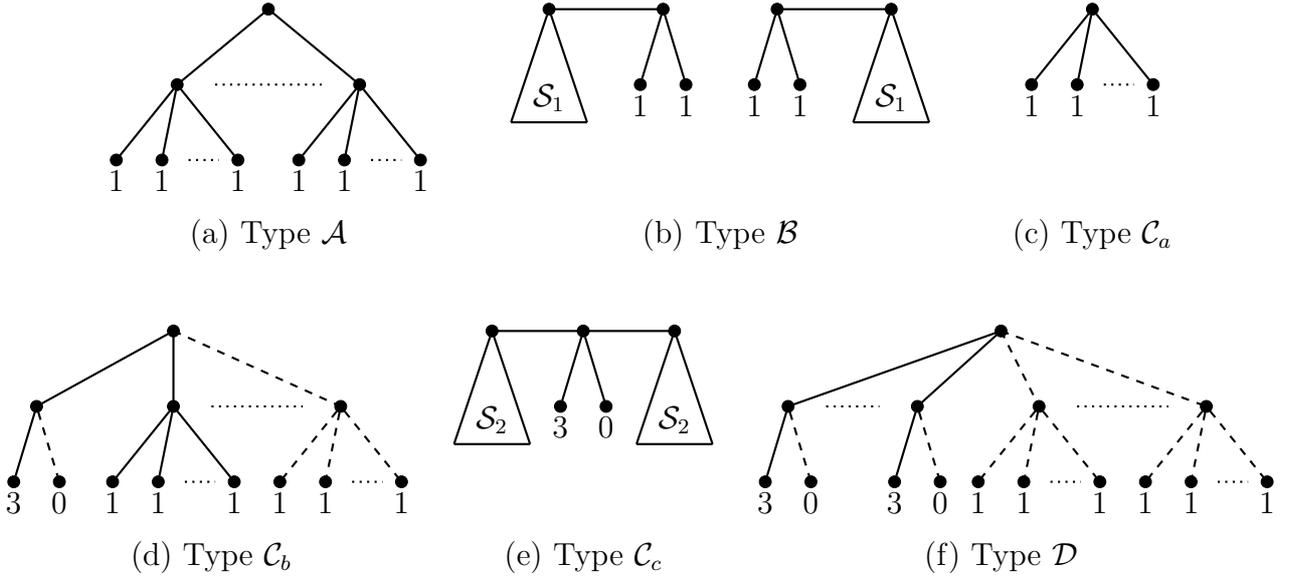
\begin{figure}
\begin{center}
%
%
\begin{tikzpicture}
\CTzero{0}{0}  \CTtroisplusL{-1.2}{-1}{1}{1}{1} \CTtroisplusL{1.2}{-1}{1}{1}{1}
\draw[thick] (0,0) to (-1.2,-1);
\draw[thick] (0,0) to (1.2,-1);
\draw[thick,dotted] (-0.7,-1) to (0.7,-1);
\node at (0,-3) {(a) Type $\FFFA$};
\end{tikzpicture}
\hskip 0.8cm
\begin{tikzpicture}
\CTtriangleL{0}{0}{$\SSS_1$}
\CTdeuxL{1.5}{0}{1}{1}
\draw[thick] (0,0) to (1.5,0);
\CTdeuxL{3}{0}{1}{1}
\CTtriangleL{4.5}{0}{$\SSS_1$}
\draw[thick] (3,0) to (4.5,0);
\node at (2.25,-3) {(b) Type $\GGG$};
\end{tikzpicture}
\hskip 0.8cm
\begin{tikzpicture}
\CTtroisplusL{0}{0}{1}{1}{1}
\node at (0,-3) {(c) Type $\XXX_a$};
\end{tikzpicture}
\vskip 0.8cm
\begin{tikzpicture}
\CTzero{0}{0}
\CTdeuxmoinsL{-1.8}{-1}{3}{0}
\CTtroisplusL{0}{-1}{1}{1}{1}
\CTtroisplusDL{2.2}{-1}{1}{1}{1}
\draw[thick] (0,0) to (-1.8,-1);
\draw[thick] (0,0) to (0,-1);
\draw[thick,dashed] (0,0) to (2.2,-1);
\draw[thick,dotted] (0.5,-1) to (1.7,-1);
\node at (0.5,-3) {(d) Type $\XXX_b$};
\end{tikzpicture}
\hskip 0.3cm
\begin{tikzpicture}
\CTtriangleL{0.3}{0}{$\SSS_2$}
\CTdeuxL{1.5}{0}{3}{0}
\CTtriangleL{2.7}{0}{$\SSS_2$}
\draw[thick] (0.3,0) to (2.7,0);
\node at (1.5,-3) {(e) Type $\XXX_c$};
\end{tikzpicture}
\hskip 0.3cm
\begin{tikzpicture}
\CTzero{0.5}{0}
\CTdeuxmoinsL{-2.3}{-1}{3}{0}
\CTdeuxmoinsL{-0.6}{-1}{3}{0}
\CTtroisplusDL{1}{-1}{1}{1}{1}
\CTtroisplusDL{3.2}{-1}{1}{1}{1}
\draw[thick] (0.5,0) to (-2.3,-1);
\draw[thick] (0.5,0) to (-0.6,-1);
\draw[thick,dashed] (0.5,0) to (1,-1);
\draw[thick,dashed] (0.5,0) to (3.2,-1);
\draw[thick,dotted] (1.5,-1) to (2.7,-1);
\draw[thick,dotted] (-1.8,-1) to (-1.1,-1);
\node at (0.5,-3) {(f) Type $\YYYC$};
\end{tikzpicture}
\caption{\label{fig:types}Spine-subtrees of given special types.}
\end{center}
\end{figure}

We now introduce the specific types of spine-subtrees,
or types of sequences of spine-subtrees, that will be used.
All these types are illustrated in Figure~\ref{fig:types}
(do not consider the depicted broadcast values yet, they will be discussed later,
in Claim~\ref{cl:f3}).

\begin{definition}[$\FFFA$, $\GGG$, $\XXX_a$, $\XXX_b$, $\XXX_c$, $\XXX$, $\YYYC$]\mbox{}\\
{\rm We define the following types of spine-subtrees.
\begin{itemize}
%
%
\item $\FFFA=\SSS_2[\alpha_2^*=0, \alpha_2\ge 2]$.\\
A spine-subtree of type $\FFFA$ is a spine-subtree of type $\SSS_2$
with at least two branches of type $\BBB_2$, all of them having at least three leaves.

\item $\GGG=\langle\SSS_1 \rangle\ \SSS_1[\lambda_1=2]\ |\ \SSS_1[\lambda_1=2]\ \langle\SSS_1 \rangle$.\\
A spine-subtree of type $\GGG$ is 
a spine-subtree of type $\SSS_1$ with two leaves,
having at least one neighbouring spine-subtree of type $\SSS_1$.

\item $\XXX_a=\SSS_1[\lambda_1\ge 3]$.\\
A spine-subtree of type $\XXX_a$ is 
a spine-subtree of type $\SSS_1$ with at least three leaves.

\item $\XXX_b=\SSS_2[\alpha_2^*=1, \alpha_2\ge 2]$.\\
A spine-subtree of type $\XXX_b$ is a spine-subtree of type $\SSS_2$
having at least two branches of type $\BBB_2$
with exactly one of them having at most two leaves.

\item $\XXX_c = \langle\SSS_2 \rangle\ \SSS_1[\lambda_1=2]\ \langle\SSS_2 \rangle$.\\
A spine-subtree of type $\XXX_c$ is a spine-subtree of type $\SSS_1$
with two leaves having two neighbouring spine-subtrees of type $\SSS_2$.

\item $\XXX = \XXX_a\ |\ \XXX_b\ |\ \XXX_c$.


\item $\YYYC = \SSS_2[\alpha_2^*\ge 2]$.\\
A spine-subtree of type $\YYYC$ is a spine-subtree of type $\SSS_2$
with at least two branches having at most two leaves.

\end{itemize}
}\end{definition}

The following observation directly follows from the previous definition,
considering the neighbouring requirements,
and will be useful later.

\begin{observation}
A spine-subtree of type $\GGG$ or $\XXX_a$ cannot have a spine-subtree of type $\XXX_c$
as a neighbouring spine-subtree.
\label{obs:GouXa-pasXc}
\end{observation}

We now claim that the set of types
$\{\FFFA,\GGG,\XXX_a,\XXX_b,\XXX_c,\YYYC\}$
induces a partition of the spine-subtrees of any locally uniform $2$-lobster
(with possibly empty parts).

\begin{proposition}
Let $\TTT=\{\FFFA,\GGG,\XXX_a,\XXX_b,\XXX_c,\YYYC\}$,
and $L$ be any locally uniform $2$-lobster.
Every spine-subtree of $L$ belongs to exactly one type in $\TTT$.
\label{prop:partition}
\end{proposition} 

\begin{proof}
Clearly, the types in $\TTT$
are pairwise disjoint, that is, no spine-subtree of $L$ can belong to 
two types from this set (see Figure~\ref{fig:types}).

We now prove that every spine-subtree of a locally uniform $2$-lobster belongs to exactly one type in this set.
Indeed, consider any such spine-subtree $S$.
\begin{enumerate}
\item
If $S$ is of type $\SSS_1$, then either $S$ has at least three leaves (type~$\XXX_a$),
or two leaves and a neighbouring spine-subtree of type $\SSS_1$ (type~$\GGG$),
or two leaves and no neighbouring spine-subtree of type $\SSS_1$ (type~$\XXX_c$).
\item
Suppose now that $S$ is of type $\SSS_2$.
Since $S$ has at least two branches, we get that
either each of these branches have at least three leaves (type~$\FFFA$),
or $S$ has exactly one branch with at most two leaves (type~$\XXX_b$),
or $S$ has at least two branches with at most two leaves (type~$\YYYC$).
\end{enumerate}
This completes the proof.
\end{proof}

\subsection{Definition of $\beta^*(L)$}
\label{sub:beta-star}

We are now able to define the value $\beta^*(L)$ for any locally uniform $2$-lobster $L$, 
which will be proven to be the optimal cost of an independent broadcast on $L$.
The value $\beta^*(L)$ will be expressed as a formula involving
the number of 1-leaves, 2-leaves and 2-only-leaves,
and the number of spine-subtrees, or sequences of spine-subtrees,
of types defined in the previous subsection, appearing in $L$.

Finally, recall that $\lambda_1(L)$, $\lambda_2(L)$ and $\lambda_2^*(L)$ denote the number
of 1-leaves, of 2-leaves and of 2-only-leaves in $L$, respectively.
We are now able to define $\beta^*(L)$.

\begin{definition}[$\beta^*(L)$]{\rm 
Let $L$ be a locally uniform $2$-lobster. We then let
$$\beta^*(L) = \nu_1(L) + \nu_2(L) + \nu_3(L) + \nu_4(L),$$
where
\begin{itemize}
\item $\nu_1(L) = \lambda_1(L) + \lambda_2(L) + \lambda_2^*(L)$
is the total number of leaves in $L$, where each 2-only-leaf is counted twice.
%
\item 
$\nu_2(L)$ is the number of branches in $L$ with at most two 2-leaves, 
that belong to a spine-subtree of type $\SSS_2$ 
(that is, of depth~2).
%
\item 
$\nu_3(L)$ is the number of spine-subtrees of type $\XXX_c$  in $L$.
%
\item 
$\nu_4(L)$ is the sum, taken over all sequences of spine-subtrees $\mathbf{S}$ in $L$ of type 
$$\langle\overline{\XXX_c}.(\varnothing|\FFFA|\GGG|\XXX_a)\rangle\ \FFFA.\{(\XXX|\FFFA).\FFFA\}^*\ \langle(\varnothing|\FFFA|\GGG|\XXX_a).\overline{\XXX_c}\rangle,$$
of the value 
$$\frac{\ell(\mathbf{S})+1}{2} - \#_{\XXX_b,\XXX_c}(\mathbf{S}),$$
where $\ell(\mathbf{S})$ denotes the number of spine-subtrees in $\mathbf{S}$,
and $\#_{\XXX_b,\XXX_c}(\mathbf{S})$ the number of spine-subtrees 
of type $\XXX_b$ or $\XXX_c$ in $\mathbf{S}$.

\end{itemize}

} 
\label{def:nu}
\end{definition}

\subsection{Lower bound}
\label{sub:lower-bound}

We will now prove that every locally uniform $2$-lobster
admits an independent broadcast $f$ with $\cost(f)=\beta^*(L)$.
We consider the case of locally uniform $2$-lobsters of length~0 separately.

\begin{lemma}
If $L$ is a locally uniform $2$-lobster of length~$k=0$, then there exists an 
independent broadcast $f$ on $L$ with $\cost(f)=\beta^*(L)$,
thus implying $\beta_b(L)\ge\beta^*(L)$.
\label{lem:beta-star-zero}
\end{lemma}

\begin{proof}
Recall that since $k=0$, $L=S_0$ is necessarily of type $\SSS_2$ (Observation~\ref{obs:ends})
and has at least two branches, so that $\diam(L)=4$.
We construct an independent broadcast $f$ on $L$ as follows,
by considering each branch separately. 
Let $B$ be any branch of $L$.
If $B$ has at most two leaves, then we set $f(\ell)=3$ for one leaf $\ell$ of $B$,
and $f(\ell')=0$ for every sister-leaf $\ell'$ of $\ell$, if any.
If $B$ has at least three leaves, then we set $f(\ell)=1$ for every leaf $\ell$ of $B$.
Finally, if $S_0$ is of type $\FFFA$, then we set $f(v_0)=1$ (in that case, $\nu_4(L)=1$).

We then have $\cost(f) = \nu_1(L) + \nu_2(L) + \nu_4(L)$,
and thus, since $\nu_3(L)=0$, $\cost(f) = \beta^*(L)$.
\end{proof}

\begin{figure}
\begin{center}
%
%
\begin{tikzpicture}
\draw[thick] (0.1,0) to (11,0);
\CTzero{0.1}{0} \CTunL{-0.3}{-1}{2}  \CTdeuxL{0.5}{-1}{1}{1}
\draw[thick] (0.1,0) to (-0.3,-1);
\draw[thick] (0.1,0) to (0.5,-1);
\CTtroisL{2}{0}{1}{1}{1}
\CTunL{3.5}{-1}{2}  \CTunL{4.1}{-1}{2} 
\CTunL{4.7}{-1}{2}  \CTdeuxL{5.6}{-1}{1}{1}   \CTtroisL{7.1}{-1}{1}{1}{1}
\CTzero{5.3}{0}
\draw[thick] (5.3,0) to (3.5,-1);
\draw[thick] (5.3,0) to (4.1,-1);
\draw[thick] (5.3,0) to (4.7,-1);
\draw[thick] (5.3,0) to (5.6,-1);
\draw[thick] (5.3,0) to (7.1,-1);
\CTdeuxL{8.6}{0}{1}{1}
\CTdeuxL{9.8}{0}{1}{1}
\CTdeux{11}{0}  \CTunL{10.7}{-1}{2}  \CTunL{11.3}{-1}{2}
\node at (5,-3) {(a) The independent broadcast $f_1$ on a sample locally uniform $2$-lobster};
\end{tikzpicture}
\vskip 0.8cm
%
%
\begin{tikzpicture}
\CTunL{4.7}{-1}{2}  \CTdeuxL{5.6}{-1}{1}{1}   \CTtroisL{7.1}{-1}{1}{1}{1}
\CTzero{5.9}{0}
\draw[thick] (5.9,0) to (4.7,-1);
\draw[thick] (5.9,0) to (5.6,-1);
\draw[thick] (5.9,0) to (7.1,-1);
\node at (9,-1) {$\longrightarrow$};
\CTunL{5.9+4.7}{-1}{3}  \CTdeuxL{5.9+5.6}{-1}{3}{0}   \CTtroisL{5.9+7.1}{-1}{1}{1}{1}
\CTzero{5.9+5.9}{0}
\draw[thick] (5.9+5.9,0) to (5.9+4.7,-1);
\draw[thick] (5.9+5.9,0) to (5.9+5.6,-1);
\draw[thick] (5.9+5.9,0) to (5.9+7.1,-1);
\node at (9,-3) {(b) From $f_1$ to $f_2$};
\end{tikzpicture}
\vskip 0.8cm
\begin{tikzpicture}
\CTdeuxL{0}{0}{1}{1}
\CTtriangleL{-1.5}{0}{$\SSS_2$} \CTtriangleL{1.5}{0}{$\SSS_2$}
\draw[thick] (-1.5,0) to (1.5,0);
\node at (3.5,-1) {$\longrightarrow$};
\CTdeuxL{7+0}{0}{3}{0}
\CTtriangleL{7-1.5}{0}{$\SSS_2$} \CTtriangleL{7+1.5}{0}{$\SSS_2$}
\draw[thick] (7-1.5,0) to (7+1.5,0);
\node at (3.5,-2.2) {(c) From $f_2$ to $f_3$};
\end{tikzpicture}
\vskip 0.8cm
\begin{tikzpicture}
\CTzero{0}{0}  \CTtroisL{-0.8}{-1}{1}{1}{1} \CTtroisL{0.8}{-1}{1}{1}{1}
\draw[thick] (0,0) to (-0.8,-1);
\draw[thick] (0,0) to (0.8,-1);
\CTdeuxL{1.9}{0}{3}{0}
\CTzero{3.8+0}{0}  \CTtroisL{3.8-0.8}{-1}{1}{1}{1} \CTtroisL{3.8+0.8}{-1}{1}{1}{1}
\draw[thick] (3.8+0,0) to (3.8-0.8,-1);
\draw[thick] (3.8+0,0) to (3.8+0.8,-1);
\CTzero{0.5+6.6}{0}  \CTtroisL{0.5+5.8}{-1}{1}{1}{1} \CTtroisL{0.5+7.4}{-1}{1}{1}{1}
\draw[thick] (0.5+6.6,0) to (0.5+5.8,-1);
\draw[thick] (0.5+6.6,0) to (0.5+7.4,-1);
\CTzero{10.4+0}{0}  \CTtroisL{10.4-0.8}{-1}{1}{1}{1} \CTtroisL{10.4+0.8}{-1}{1}{1}{1}
\draw[thick] (10.4+0,0) to (10.4-0.8,-1);
\draw[thick] (10.4+0,0) to (10.4+0.8,-1);
\CTdeuxL{12.3}{0}{1}{1}
\CTdeuxL{13.3}{0}{1}{1}
\draw[thick] (0,0) to (13.3,0);
\end{tikzpicture}
\vskip 0.5cm
\begin{tikzpicture}
\node at (-1.4,-0.6) {$\longrightarrow$};
\CTzeroL{0}{0}{1}  \CTtroisL{-0.8}{-1}{1}{1}{1} \CTtroisL{0.8}{-1}{1}{1}{1}
\draw[thick] (0,0) to (-0.8,-1);
\draw[thick] (0,0) to (0.8,-1);
\CTdeuxL{1.9}{0}{1}{1}
\CTzeroL{3.8+0}{0}{1}  \CTtroisL{3.8-0.8}{-1}{1}{1}{1} \CTtroisL{3.8+0.8}{-1}{1}{1}{1}
\draw[thick] (3.8+0,0) to (3.8-0.8,-1);
\draw[thick] (3.8+0,0) to (3.8+0.8,-1);
\CTzero{0.5+6.6}{0}  \CTtroisL{0.5+5.8}{-1}{1}{1}{1} \CTtroisL{0.5+7.4}{-1}{1}{1}{1}
\draw[thick] (0.5+6.6,0) to (0.5+5.8,-1);
\draw[thick] (0.5+6.6,0) to (0.5+7.4,-1);
\CTzeroL{10.4+0}{0}{1}  \CTtroisL{10.4-0.8}{-1}{1}{1}{1} \CTtroisL{10.4+0.8}{-1}{1}{1}{1}
\draw[thick] (10.4+0,0) to (10.4-0.8,-1);
\draw[thick] (10.4+0,0) to (10.4+0.8,-1);
\CTdeuxL{12.3}{0}{1}{1}
\CTdeuxL{13.3}{0}{1}{1}
\draw[thick] (0,0) to (13.3,0);
\node at (6.3,-3) {(d) From $f_3$ to $f_4$, for a sequence of type
$\langle\varnothing\rangle\FFFA.\XXX_c.\FFFA.\FFFA.\FFFA\langle\GGG.\GGG\rangle$
(cost increases by~2)};
\end{tikzpicture}
\caption{\label{fig:beta-star-a}Proof of Lemma~\ref{lem:beta-star-general}: from $f_1$ to $f_4$.}
\end{center}
\end{figure}

\begin{lemma}
Every locally uniform $2$-lobster $L$ of length $k\ge 1$
admits an independent broadcast $f$ with $\cost(f)=\beta^*(L)$,
thus implying $\beta_b(L)\ge\beta^*(L)$.
\label{lem:beta-star-general}
\end{lemma}

\begin{proof}
We will construct a sequence of four independent broadcasts $f_1,\dots,f_4$ on $L$,
step by step, such that $\cost(f_4)=\beta^*(L)$.
Each independent broadcast $f_i$, $2\le i\le 4$, is obtained by possibly modifying 
the independent broadcast $f_{i-1}$, and is such that $\cost(f_i)\ge\cost(f_{i-1})$.
Moreover, for each independent broadcast $f_i$, $1\le i\le 3$, we will have $f_i(v_j)=0$
for every spine-vertex $v_j$, $0\le j\le k$, while we may have $f_4(v_j)=1$.

These modifications are illustrated in Figures~\ref{fig:beta-star-a},
where dashed edges represent optional edges.
These figures should help the reader to see that each mapping $f_i$ is a valid independent broadcast
on $L$.

\bigskip

\noindent{\em Step 1.}
Let $f_1$ be the mapping defined by
$f_1(u)=2$ if $u$ is an only-leaf, 
$f_1(u)=1$ if $u$ is a leaf which is not an only-leaf
and $f_1(u)=0$ otherwise
(see Figure~\ref{fig:beta-star-a}(a)).

Clearly, $f_1$ is an independent broadcast on $L$ with 
$$\cost(f_1)= \lambda_1(L) + \lambda_2(L) + \lambda_2^*(L) =\nu_1(L).$$

\bigskip

\noindent{\em Step 2.}
We modify $f_1$ as follows, to obtain $f_2$.
For every branch $B_i^j$ of type $\BBB_2[\lambda_2\le 2]$, $0\le i\le k$, $1\le j\le s_i$,
such that $S_i$ is a spine-subtree of type $\SSS_2$,
we let $f_2(\ell)=3$ for one leaf $\ell$ of $B_i^j$, and $f_2(\ell')=0$ for the sister-leaf $\ell'$
of $\ell$, if any
(see Figure~\ref{fig:beta-star-a}(b)).

Again, $f_2$ is an independent broadcast on $L$ with 
$$\cost(f_2)=\cost(f_1) + \nu_2(L).$$

\bigskip

\noindent{\em Step 3.}
We modify $f_2$ as follows, to obtain $f_3$.
For every spine-subtree $S$ of type $\XXX_c$,
we let
$f_3(\ell)=3$ for one leaf $\ell$ of $S$, and $f_3(\ell')=0$ for the sister-leaf $\ell'$ of $\ell$, if any
(see Figure~\ref{fig:beta-star-a}(c)).
Note here that setting $f_3(\ell)=3$ is allowed, since $L$ does not contain any
vertex $x$ with $f_3(x)>0$ at distance at most~3 from the leaves of $S$,
and $f_3(\ell'')\le 3$ for every 2-leaf $\ell''$ of $L$.

Again, $f_3$ is an independent broadcast on $L$, and,
since the broadcast value of every considered spine-subtree $S$ has been increased by~1, we have 
$$\cost(f_3)=\cost(f_2) + \nu_3(L).$$

\bigskip

Before describing the last step, we prove a claim on the $f_3$-values
and introduce some terminology.

\begin{claim}
\label{cl:f3}
After step~3, the $f_3$-values of the vertices of $L$, depending on the
type of the spine-subtree they belong to, are those values depicted in Figure~\ref{fig:types}.
\end{claim}

\begin{proof}
Let $S$ be any spine-subtree of $L$.
If $S$ is of type $\SSS_1$, then either $S$ has at least three leaves (type~$\XXX_a$), 
or two leaves and a neighbouring spine-subtree of type $\SSS_1$ (type~$\GGG$),
and has thus not been modified in steps 2 or~3 in both cases,
or two leaves and no neighbouring spine-subtree of type $\SSS_1$ (type~$\XXX_c$),
and has thus been modified in step~3.
In all these cases, the $f_3$-values of the vertices in $S$ are those values depicted
in Figure~\ref{fig:types}(b), (c) or (e).

Suppose now that $S$ is of type $\SSS_2$.
In that case, either each of these branches have at least three leaves (type~$\FFFA$),
and has thus not been modified in steps 2 or~3,
or $S$ has exactly one branch with at most two leaves (type~$\XXX_b$),
or $S$ has at least two branches with at most two leaves (type~$\YYYC$),
and has thus been modified in step~2.
In all these cases, the $f_3$-values of the vertices in $S$ are those values depicted
in Figure~\ref{fig:types}(a), (d) or (f).

This completes the proof of the claim.
\end{proof}

A spine-subtree $S$ \emph{exceeds by $e$}, for some integer $e\ge 1$,
if $S$ contains a 1-leaf with broadcast value $e+1$, or a 2-leaf with broadcast value $e+2$.
Therefore, if a spine-subtree $S_i$, $0\le i\le k$, of a locally
uniform $2$-lobster $L$ exceeds by $e\ge 1$, then none of the spine vertices
$v_{i-e},\dots,v_{i},\dots,v_{i+e}$ can be a broadcast vertex.

We can then partition the set 
$\TTT=\{\FFFA,\GGG,\XXX_a,\XXX_b,\XXX_c,\YYYC\}$
in three parts $\EEE_0$, $\EEE_1$ and $\EEE_2$, corresponding to the types of spine-subtrees exceeding by~0,
1 and 2, respectively, after Step~3 (as given in Claim~\ref{cl:f3}). 
In order to be complete, we will also say that the ``empty subtree'', of type $\varnothing$,
does not exceed. Therefore, we have
$$\EEE_0 = \{\varnothing, \FFFA, \GGG, \XXX_a \},\ \ 
\EEE_1 = \{ \XXX_b, \YYYC \}\ \ \mbox{and}\ \ 
\EEE_2 = \{ \XXX_c \}.$$
Moreover, we denote by $\overline{\EEE_i}$ the complement of $\EEE_i$ for every $i$,
$0\le i\le 2$, that is, $\overline{\EEE_i} = (\TTT \cup \{\varnothing\}) \setminus \EEE_i$.

Let $S$ be a spine-subtree of type $\FFFA$.
By \emph{increasing $S$ by one}, we mean giving the broadcast value~$1$ to the root
of $S$ (observe that only leaves of $S$ are $f_3$-broadcast vertices, and that
$f_3(\ell)=1$ for every such leaf~$\ell$).

Let now $S$ be a spine-subtree of type $\XXX_b$ or $\XXX_c$.
By \emph{decreasing $S$ by one}, we mean the following:
\begin{itemize}
\item If $S$ is of type $\XXX_b$, then we give the broadcast value~$2$ to one
leaf of the (unique) branch of type $\BBB_2[\lambda_2\le 2]$, and the broadcast value~$0$ to its
sister-leaf, if any (by Claim~\ref{cl:f3}, $f_3(\ell)=3$ for one leaf $\ell$ of $S$,
and $f_3(\ell')=0$ for the sister-leaf $\ell'$ of $\ell$, if any).
\item If $S$ is of type $\XXX_c$, then we give the broadcast value~$1$ to each
of the two leaves of $S$ (by Claim~\ref{cl:f3}, $f_3(\ell)=3$ for one leaf $\ell$ of $S$,
and $f_3(\ell')=0$ for the sister-leaf $\ell'$ of $\ell$).
\end{itemize}

Observe that after having being decreased by one, a spine-subtree of type
$\XXX_b$ or~$\XXX_c$ does no longer exceed.

We are now able to describe the fourth step of the proof.
The key idea of this last step is to increase by one
some spine-subtrees of type $\FFFA$, and decrease by one
some spine-subtrees of type $\XXX_b$ or~$\XXX_c$,
provided that this results in a strict increasing of
the cost of the current independent broadcast on $L$.

\bigskip

\noindent{\em Step 4.}   
We modify $f_3$ as follows, to obtain $f_4$.
For every sequence of spine-subtrees $A_0X_1A_1\dots X_pA_p$, $p\ge 0$,
of type
$$\TTT_4 = \langle\overline{\EEE_2}.\EEE_0\rangle\ \FFFA.\{(\XXX|\FFFA).\FFFA\}^*\ \langle\EEE_0.\overline{\EEE_2}\rangle,$$ 
we decrease by one each spine-subtree $X_i$ of type $\XXX_b$ or $\XXX_c$,
$1\le i\le p$, and increase by one each spine-subtree $A_j$, $0\le j\le p$
(see Figure~\ref{fig:beta-star-a}(d)).
Note that this can be done since none of the spine-subtrees $X_i$, $1\le i\le p$, exceeds
and no spine-subtree outside the sequence can prevent us from doing so on
the extremal spine-subtrees $A_0$ and $A_p$.

The broadcast value of the whole sequence is thus increased by~$p+1$, 
minus the number of spine-subtrees of type $\XXX_b$ or $\XXX_c$.
Since the number of spine-subtrees of type $\XXX_b$ or $\XXX_c$
is at most~$p$, this broadcast value always increases.
Therefore, doing the above modification for every sequence of spine-subtrees
of type $\TTT_4$, the so-obtained independent broadcast $f_4$ satisfies 
$$\cost(f_4) = \cost(f_3) + \nu_4(L).$$

We finally get $\cost(f_4)=\beta^*(L)$, as required. This completes the proof.
\end{proof}

\subsection{Upper bound}
\label{sub:upper-bound}

We first prove that, for every locally uniform $2$-lobster $L$,
we can choose an optimal independent broadcast on $L$ that satisfies
some given properties.


The next lemma shows that if $f$ is an optimal independent broadcast
on a locally uniform $2$-lobster $L$ of length $k\ge 1$,
then 
there exists an optimal independent broadcast $\tf$ on $L$
such that the $\tf$-values of the vertices in the spine-subtrees $S_0$ and $S_k$
are at most~$3$.

\begin{lemma}
If $L$ is a locally uniform $2$-lobster of length $k\ge 1$, 
and $f$ is an optimal independent broadcast on $L$, 
then 
there exists an optimal independent broadcast $\tf$ on $L$
such that $\tf(\ell)\leq 3$ for every leaf $\ell$ of $S_0$ and $S_k$.
\label{lem:end-spine-subtrees}
\end{lemma}

\begin{proof} 
Recall first that, by Observation~\ref{obs:ends}, both spine-subtrees $S_0$ and $S_k$ must be of type $\SSS_2$.
Also note that, by symmetry, it is enough to prove the result for~$S_0$.
If $S_0$ has at least two broadcast leaves, then the broadcast value of each of them is
at most~3, since every two such leaves are at distance at most~4 from each other.
We thus only need to consider the case when $S_0$ has a unique broadcast leaf.
Moreover, we can assume that the broadcast value of this leaf is at least~7,
since otherwise, by setting the broadcast value of any two leaves at distance~4 from
each other to~3, we would get either a broadcast satisfying the requirement of the lemma,
or a contradiction with the optimality of the broadcast.
Therefore, we get that the result holds if $k\le 3$ since, in that case,
$\diam(L)\le 7$, which implies $f(v)\le 6$ for every vertex $v$ of $L$
since $f$ is maximal.  

The proof now is by contradiction.
Let $L$ be a counter-example to the lemma, of length $k\ge 4$, 
and $f$ be an optimal independent broadcast on $L$ which minimizes
the value of $f(\ell)=\alpha$,
where $\ell$ is the (unique) $f$-broadcast leaf of $S_0$. We thus have $\alpha\geq 7$.

Observe that at least one vertex at distance $\alpha+1$ from $\ell$
must be an $f$-broadcast vertex, since otherwise we could increase
the value of $f(\ell)$ by~1, contradicting the optimality of $f$.
Let $x$ denote any such vertex.
The spine-subtrees $S_1,\dots,S_{\alpha-4}$ do not contain any
$f$-broadcast vertex (since every such vertex is $f$-dominated by $\ell$), and 
$x$ is either a 2-leaf of $S_{\alpha-3}$, 
a 1-leaf of $S_{\alpha-2}$,
or the spine-vertex $v_{\alpha-1}$.

We consider four cases, depending on whether these vertices are $f$-broadcast vertices or not.
For each of these cases, we assume that none of the previous cases occurs.

\begin{enumerate}

\item {\it $v_{\alpha-1}$ is an $f$-broadcast vertex.}\\
In this case, by Corollary~\ref{cor:lobster}, we know that $f(x)=1$.
Consider the spine-subtree $S_{\alpha-3}$.
If $S_{\alpha-3}$ is of type $\SSS_1$, then all its vertices are
$f$-dominated by $\ell$, and the mapping $g$ defined by $g(\ell)=\alpha-1$,
$g(\ell_{\alpha-3})=2$ for one leaf $\ell_{\alpha-3}$ of $S_{\alpha-3}$
and $g(v)=f(v)$ for every other vertex $v$ of $L$ is clearly an independent
broadcast on $L$ with $\cost(g)=\cost(f)-1+2=\cost(f)+1$, which contradicts the optimality of $f$
(see Figure~\ref{fig:bords-cas-1-2}(a)).
Now, if $S_{\alpha-3}$ is of type $\SSS_2$, then
$f(\ell_{\alpha-3})\le 3$ for every leaf $\ell_{\alpha-3}$ of $S_{\alpha-3}$.
Therefore, the mapping $g$ defined by $g(\ell)=\alpha-2$,
$g(\ell_{\alpha-4})=2$ for one leaf $\ell_{\alpha-4}$ of $S_{\alpha-4}$
and $g(v)=f(v)$ for every other vertex $v$ of $L$ is clearly an independent
broadcast on $L$, with $\cost(g)=\cost(f)-2+2=\cost(f)$, which contradicts the minimality of $\alpha$
(see Figure~\ref{fig:bords-cas-1-2}(b), where $S_{\alpha-4}$ is supposed to be of type $\SSS_1$,
the case $S_{\alpha-4}$ of type $\SSS_2$ being similar).
Therefore, $v_{\alpha-1}$ cannot be an $f$-broadcast vertex.

\item {\it Both a 2-leaf $\ell_{\alpha-3}$ of $S_{\alpha-3}$ and
a 1-leaf $\ell_{\alpha-2}$ of $S_{\alpha-2}$ are $f$-broadcast vertices.}\\
In this case, we necessarily have $f(\ell_{\alpha-3})\leq 3$ and $f(\ell_{\alpha-2})\leq 3$.
Therefore, the mapping $g$ defined by $g(\ell)=\alpha-2$,
$g(\ell_{\alpha-4})=3$ for one leaf $\ell_{\alpha-4}$ of $S_{\alpha-4}$
and $g(v)=f(v)$ for every other vertex $v$ of $L$ is clearly an independent
broadcast on $L$ with $\cost(g)=\cost(f)-2+3=\cost(f)+1$, which contradicts the optimality of $f$
(see Figure~\ref{fig:bords-cas-1-2}(c), where again $S_{\alpha-4}$ is supposed to be of type $\SSS_1$,
the case $S_{\alpha-4}$ of type $\SSS_2$ being similar).

\begin{figure}
\begin{center}

\begin{tikzpicture}
\node at (0,0.4) {$S_{0}$};
\node at (2.2,0.4) {$S_{\alpha-3}$};
\node at (3.4,0.4) {$S_{\alpha-2}$};
%
\CTun{0}{0}  \CTunL{2.2}{0}{0}  \CTunL{3.4}{0}{$\le 1$}   \CTzeroL{4.6}{0}{1}
\CTunL{0}{-1}{$\alpha$}
\draw[thick,dotted] (0,0) to (2.2,0);
\draw[thick] (0,0) to (0.3,0);
\draw[thick] (1.9,0) to (4.6,0);
\node at (6,-1) {$\longrightarrow$};
%
\node at (7.4+0,0.4) {$S_{0}$};
\node at (7.4+2.2,0.4) {$S_{\alpha-3}$};
\node at (7.4+3.4,0.4) {$S_{\alpha-2}$};
%
\CTun{7.4+0}{0}  \CTunL{7.4+2.2}{0}{2}  \CTunL{7.4+3.4}{0}{$\le 1$}   \CTzeroL{7.4+4.6}{0}{1}
\CTunL{7.4+0}{-1}{$\alpha-1$}
\draw[thick,dotted] (7.4+0,0) to (7.4+2.2,0);
\draw[thick] (7.4+0,0) to (7.4+0.3,0);
\draw[thick] (7.4+1.9,0) to (7.4+4.6,0);
\node at (6,-3) {(a) Case 1, $S_{\alpha-2}$ is of type $\SSS_1$};
\end{tikzpicture}
\vskip 0.5cm
\begin{tikzpicture}
\node at (0,0.4) {$S_{0}$};
\node at (2.2,0.4) {$S_{\alpha-4}$};
\node at (3.4,0.4) {$S_{\alpha-3}$};
\node at (4.6,0.4) {$S_{\alpha-2}$};
%
\CTun{0}{0}  \CTunL{2.2}{0}{0} \CTun{3.4}{0}  \CTunL{4.6}{0}{$\le 1$}   \CTzeroL{5.8}{0}{1}
\CTunL{0}{-1}{$\alpha$}   \CTunL{3.4}{-1}{$\le 3$}      
\draw[thick,dotted] (0,0) to (2.2,0);
\draw[thick] (0,0) to (0.3,0);
\draw[thick] (1.9,0) to (5.8,0);
\node at (7.2,-1) {$\longrightarrow$};
%
\node at (8.6+0,0.4) {$S_{0}$};
\node at (8.6+2.2,0.4) {$S_{\alpha-4}$};
\node at (8.6+3.4,0.4) {$S_{\alpha-3}$};
\node at (8.6+4.6,0.4) {$S_{\alpha-2}$};
%
\CTun{8.6+0}{0}  \CTunL{8.6+2.2}{0}{2} \CTun{8.6+3.4}{0}  \CTunL{8.6+4.6}{0}{$\le 1$}   \CTzeroL{8.6+5.8}{0}{1}
\CTunL{8.6+0}{-1}{$\alpha-2$}   \CTunL{8.6+3.4}{-1}{$\le 3$}      
\draw[thick,dotted] (8.6+0,0) to (8.6+2.2,0);
\draw[thick] (8.6+0,0) to (8.6+0.3,0);
\draw[thick] (8.6+1.9,0) to (8.6+5.8,0);
\node at (7.2,-3) {(b) Case 1, $S_{\alpha-2}$ is of type $\SSS_2$};
\end{tikzpicture}
\vskip 0.5cm
\begin{tikzpicture}
\node at (0,0.4) {$S_{0}$};
\node at (2.2,0.4) {$S_{\alpha-4}$};
\node at (3.4,0.4) {$S_{\alpha-3}$};
\node at (4.6,0.4) {$S_{\alpha-2}$};
\CTun{0}{0}  \CTunL{2.2}{0}{0}  \CTun{3.4}{0}   \CTunL{4.6}{0}{$\le 3$}
\CTunL{0}{-1}{$\alpha$}    \CTunL{3.4}{-1}{$\le 3$}
\draw[thick,dotted] (0,0) to (2.2,0);
\draw[thick] (0,0) to (0.3,0);
\draw[thick] (1.9,0) to (4.6,0);
\node at (6,-1) {$\longrightarrow$};
%
\node at (7.4+0,0.4) {$S_{0}$};
\node at (7.4+2.2,0.4) {$S_{\alpha-4}$};
\node at (7.4+3.4,0.4) {$S_{\alpha-3}$};
\node at (7.4+4.6,0.4) {$S_{\alpha-2}$};
\CTun{7.4+0}{0}  \CTunL{7.4+2.2}{0}{3}  \CTun{7.4+3.4}{0}   \CTunL{7.4+4.6}{0}{$\le 3$}
\CTunL{7.4+0}{-1}{$\alpha-2$}    \CTunL{7.4+3.4}{-1}{$\le 3$}
\draw[thick,dotted] (7.4+0,0) to (7.4+2.2,0);
\draw[thick] (7.4+0,0) to (7.4+0.3,0);
\draw[thick] (7.4+1.9,0) to (7.4+4.6,0);
\node at (6,-3) {(c) Case 2};
\end{tikzpicture}
\caption{\label{fig:bords-cas-1-2}Independent broadcasts for the proof of Lemma~\ref{lem:end-spine-subtrees}, case 1 and case 2 (only one branch per spine-subtree is depicted).
}
\end{center}

\end{figure}

\begin{figure}
\begin{center}
\begin{tikzpicture}
\node at (0,0.4) {$S_{0}$};
\node at (2.2,0.4) {$S_{\alpha-3}$};
\node at (3.4,0.4) {$S_{\alpha-2}$};
%
\CTun{0}{0}  \CTunL{2.2}{0}{0}  \CTunL{3.4}{0}{1}   
\CTunL{0}{-1}{$\alpha$}
\draw[thick,dotted] (0,0) to (2.2,0);
\draw[thick] (0,0) to (0.3,0);
\draw[thick] (1.9,0) to (3.4,0);
\node at (4.8,-1) {$\longrightarrow$};
%
\node at (6.2+0,0.4) {$S_{0}$};
\node at (6.2+2.2,0.4) {$S_{\alpha-3}$};
\node at (6.2+3.4,0.4) {$S_{\alpha-2}$};
%
\CTun{6.2+0}{0}  \CTunL{6.2+2.2}{0}{2}  \CTunL{6.2+3.4}{0}{1}   
\CTunL{6.2+0}{-1}{$\alpha-1$}
\draw[thick,dotted] (6.2+0,0) to (6.2+2.2,0);
\draw[thick] (6.2+0,0) to (6.2+0.3,0);
\draw[thick] (6.2+1.9,0) to (6.2+3.4,0);
\node at (4.8,-3) {(a) Case 3, $\beta=1$};
\end{tikzpicture}
\vskip 0.5cm
\begin{tikzpicture}
\node at (-1.2,0.4) {$S_{0}$};
\node at (0,0.4) {$S_{1}$};
\node at (2.2,0.4) {$S_{\alpha-3}$};
\node at (3.4,0.4) {$S_{\alpha-2}$};
%
\CTun{-1.2}{0} \CTunL{0}{0}{0}  \CTunL{2.2}{0}{0}  \CTunL{3.4}{0}{$2\le\beta<\alpha$}   
\CTunL{-1.2}{-1}{$\alpha$}
\draw[thick,dotted] (-1.2,0) to (2.2,0);
\draw[thick] (-1.2,0) to (0.3,0);
\draw[thick] (1.9,0) to (3.4,0);
\node at (4.8,-1) {$\longrightarrow$};
%
\node at (7.4-1.2,0.4) {$S_{0}$};
\node at (7.4+0,0.4) {$S_{1}$};
\node at (7.4+2.2,0.4) {$S_{\alpha-3}$};
\node at (7.4+3.4,0.4) {$S_{\alpha-2}$};
%
\CTun{7.4-1.2}{0} \CTunL{7.4+0}{0}{2}  \CTunL{7.4+2.2}{0}{2}  \CTunL{7.4+3.4}{0}{2}   
\CTunL{7.4-1.2}{-1}{$3$}
\draw[thick,dotted] (7.4-1.2,0) to (7.4+2.2,0);
\draw[thick] (7.4-1.2,0) to (7.4+0.3,0);
\draw[thick] (7.4+1.9,0) to (7.4+3.4,0);
\node at (4.8,-3) {(b) Case 3, $2\le \beta<\alpha$};
\end{tikzpicture}
\vskip 0.5cm
\begin{tikzpicture}
\node at (-1.2,0.4) {$S_{0}$};
\node at (0,0.4) {$S_{1}$};
\node at (2.2,0.4) {$S_{\alpha-3}$};
\node at (3.4,0.4) {$S_{\alpha-2}$};
\node at (4.6,0.4) {$S_{\alpha-1}$};
\CTun{-1.2}{0} \CTunL{0}{0}{0}  \CTunL{2.2}{0}{0}  \CTunL{3.4}{0}{$\alpha$}  \CTunL{4.6}{0}{0}    
\CTunL{-1.2}{-1}{$\alpha$}
\draw[thick,dotted] (-1.2,0) to (2.2,0);
\draw[thick] (-1.2,0) to (0.3,0);
\draw[thick] (1.9,0) to (4.6,0);
\node at (6,-1) {$\longrightarrow$};
%
\node at (8.6-1.2,0.4) {$S_{0}$};
\node at (8.6+0,0.4) {$S_{1}$};
\node at (8.6+2.2,0.4) {$S_{\alpha-3}$};
\node at (8.6+3.4,0.4) {$S_{\alpha-2}$};
\node at (8.6+4.6,0.4) {$S_{\alpha-1}$};
\CTun{8.6-1.2}{0} \CTunL{8.6+0}{0}{2}  \CTunL{8.6+2.2}{0}{2}  \CTunL{8.6+3.4}{0}{2}  \CTunL{8.6+4.6}{0}{2}    
\CTunL{8.6-1.2}{-1}{$3$}
\draw[thick,dotted] (8.6-1.2,0) to (8.6+2.2,0);
\draw[thick] (8.6-1.2,0) to (8.6+0.3,0);
\draw[thick] (8.6+1.9,0) to (8.6+4.6,0);
\node at (6,-3) {(c) Case 3, $\beta=\alpha$ and $\ell_{\alpha-1}$ is $f$-dominated only by $\ell_{\alpha-2}$};
\end{tikzpicture}
\vskip 0.5cm
\begin{tikzpicture}
\node at (-1.2,0.4) {$S_{0}$};
\node at (0,0.4) {$S_{1}$};
\node at (1.8,0.4) {$S_{\alpha-2}$};
\node at (3,0.4) {$S_{\alpha-1}$};
\node at (4.8,0.4) {$S_{2\alpha-3}$};
\CTun{-1.2}{0} \CTunL{0}{0}{0}  \CTunL{1.8}{0}{$\alpha$}  \CTunL{3}{0}{0}  \CTunL{4.8}{0}{$\alpha$}    
\CTunL{-1.2}{-1}{$\alpha$}
\draw[thick,dotted] (-1.2,0) to (1.8,0);
\draw[thick,dotted] (3,0) to (4.8,0);
\draw[thick] (-1.2,0) to (0.3,0);
\draw[thick] (1.5,0) to (3.3,0);
\draw[thick] (4.5,0) to (4.8,0);
\node at (6.2,-1) {$\longrightarrow$};
\node at (8.8-1.2,0.4) {$S_{0}$};
\node at (8.8+0,0.4) {$S_{1}$};
\node at (8.8+1.8,0.4) {$S_{\alpha-3}$};
\node at (8.8+3,0.4) {$S_{\alpha-2}$};
\node at (8.8+4.8,0.4) {$S_{2\alpha-3}$};
\CTun{8.8-1.2}{0} \CTunL{8.8+0}{0}{2}  \CTunL{8.8+1.8}{0}{$2$}  \CTunL{8.8+3}{0}{2}  \CTunL{8.8+4.8}{0}{$\alpha-1$}    
\CTunL{8.8-1.2}{-1}{$3$}
\draw[thick,dotted] (8.8-1.2,0) to (8.8+1.8,0);
\draw[thick,dotted] (8.8+3,0) to (8.8+4.8,0);
\draw[thick] (8.8-1.2,0) to (8.8+0.3,0);
\draw[thick] (8.8+1.5,0) to (8.8+3.3,0);
\draw[thick] (8.8+4.5,0) to (8.8+4.8,0);
\node at (6.2,-3) {(d) Case 3, $\beta=\alpha$ and $\ell_{\alpha-1}$ is $f$-dominated by $\ell_{\alpha-2}$ and $x=\ell_{2\alpha-3}$};
\end{tikzpicture}

\caption{\label{fig:bords-cas-3}Independent broadcasts for the proof of Lemma~\ref{lem:end-spine-subtrees}, case 3 (only one branch per spine-subtree is depicted).
}
\end{center}

\end{figure}

\item {\it A 1-leaf $\ell_{\alpha-2}$ of $S_{\alpha-2}$ is an $f$-broadcast vertex.}\\
We let $\beta=f(\ell_{\alpha-2})$. We necessarily have $\beta\leq\alpha$.
If $\beta=1$, the mapping $g$ defined by $g(\ell)=\alpha-1$,
$g(\ell_{\alpha-3})=2$ for one leaf $\ell_{\alpha-3}$ of $S_{\alpha-3}$
and $g(v)=f(v)$ for every other vertex $v$ of $L$, is clearly an independent
broadcast on $L$ with $\cost(g)=\cost(f)+1$, which contradicts the optimality of $f$
(see Figure~\ref{fig:bords-cas-3}(a), where $S_{\alpha-3}$ is supposed to be of type $\SSS_1$,
the case $S_{\alpha-3}$ of type $\SSS_2$ being similar).

Suppose now $\beta\geq 2$ and let $g_{\alpha-2}$ be the mapping defined 
by $g_{\alpha-2}(\ell)=3$, $g_{\alpha-2}(\ell_j)=2$ for one leaf $\ell_j$
of each spine-subtree $S_j$, $1\le j\le \alpha-3$, $g_{\alpha-2}(\ell_{\alpha-2})=2$
and $g(v)=f(v)$ for every other vertex $v$ of $L$.
The mapping $g_{\alpha-2}$ is clearly an independent
broadcast on $L$, with 
$$\cost(g_{\alpha-2})=\cost(f)-\alpha-\beta+3+2(\alpha-2)=\cost(f)+(\alpha-\beta) -1$$
(see Figure~\ref{fig:bords-cas-3}(b), where $S_1,\dots,S_{\alpha-3}$ are supposed to be of type $\SSS_1$,
all other cases being similar).
Therefore, the mapping $g_{\alpha-2}$ contradicts either the optimality of $f$ or
the minimality of $\alpha$ if $\beta<\alpha$.

We can thus assume $\beta=\alpha$.
Suppose first that the spine-subtree $S_{\alpha-1}$ contains a leaf $\ell_{\alpha-1}$
which is $f$-dominated only by $\ell_{\alpha-2}$.
(Observe that this is in particular the case if $S_{\alpha-1}$ is of type $\SSS_2$.)
In that case, we consider the mapping $g_{\alpha-1}$ whose definition is similar
to the above definition of $g_{\alpha-2}$, by simply replacing $\alpha-2$ by $\alpha-1$.
The mapping $g_{\alpha-1}$ is clearly an independent
broadcast on $L$, with 
$$\cost(g_{\alpha-1})=\cost(f)-2\alpha +3 +2(\alpha-1)=\cost(f)+1,$$
which contradicts the optimality of $f$
(see Figure~\ref{fig:bords-cas-3}(c), where $S_1,\dots,S_{\alpha-3}$ and $S_{\alpha-1}$ 
are supposed to be of type $\SSS_1$, all other cases being similar). 

Therefore, $S_{\alpha-1}$ is of type $\SSS_1$ and each of its 1-leaves 
is $f$-dominated by $\ell_{\alpha-2}$ and (at least) one other vertex $x$.
Moreover, we necessarily have $f(x)=f(\ell_{\alpha-2})=\alpha\geq 7$, which
implies the uniqueness of~$x$,
since a 2-leaf of $S_{2\alpha-4}$ and a 1-leaf of $S_{2\alpha-3}$ are at distance~4 from each other,
and $v_{2\alpha-2}$ cannot be an $f$-broadcast vertex by Corollary~\ref{cor:lobster}.
We then consider the mapping $g'_{\alpha-1}$ 
defined by $g'_{\alpha-1}(x)=\alpha-1$
and $g'_{\alpha-1}(v)=g_{\alpha-1}(v)$ for every other vertex $v$ of $L$.
Again, the mapping $g'_{\alpha-1}$ is clearly an independent
broadcast on $L$, with 
$$\cost(g'_{\alpha-1})=\cost(g_{\alpha-1})-1=\cost(f),$$
which contradicts the minimality of $\alpha$
(see Figure~\ref{fig:bords-cas-3}(d), where $x$ is a 1-leaf of $S_{2\alpha-3}$, the case
when $x$ is a 2-leaf of $S_{2\alpha-4}$ being similar).

\item {\it A 2-leaf $\ell_{\alpha-3}$ of $S_{\alpha-3}$ is an $f$-broadcast vertex.}\\
Note first that if $\alpha-3=k$, then the optimality of $f$ implies 
$f(\ell_{\alpha-3})=\alpha$, so that $\cost(f)=2\alpha=2(\diam(L)-1)$, in
contradiction with Observation~\ref{obs:lower-bound-2-lul}. We thus have $\alpha-3<k$.

We let $\beta=f(\ell_{\alpha-3})$. We necessarily have $\beta\leq\alpha$.

If $\beta\leq 2$, the mapping $g$ defined by $g(\ell)=\alpha-1$,
$g(\ell_{\alpha-3})=3$,
and $g(v)=f(v)$ for every other vertex $v$ of $L$, is clearly an independent
broadcast on $L$ with $\cost(g)=\cost(f)-1-\beta+3=\cost(f)-\beta+2\ge\cost(f)$, 
which contradicts the minimality of $\alpha$ or the optimality of $f$
(see Figure~\ref{fig:bords-cas-4}(a)).

Suppose now $\beta\geq 3$ and let $g_{\alpha-3}$ be the mapping defined 
by $g_{\alpha-3}(\ell)=3$, $g_{\alpha-3}(\ell_j)=2$ for one leaf $\ell_j$
of each spine-subtree $S_j$, $1\le j\le \alpha-4$, $g_{\alpha-3}(\ell_{\alpha-3})=3$
and $g(v)=f(v)$ for every other vertex $v$ of $L$.
The mapping $g_{\alpha-3}$ is clearly an independent
broadcast on $L$, with 
$$\cost(g_{\alpha-3})=\cost(f)-\alpha-\beta+3+2(\alpha-4)+3=\cost(f)+(\alpha-\beta) -2$$
(see Figure~\ref{fig:bords-cas-4}(b), where $S_1,\dots,S_{\alpha-4}$ are supposed to be of type $\SSS_1$,
all other cases being similar). 
Therefore, the mapping $g_{\alpha-3}$ contradicts either the optimality of $f$ or
the minimality of $\alpha$ if $\beta<\alpha-1$.

We can thus assume $\alpha-1 \leq \beta \leq \alpha$, so that $\beta\ge 6$.
Suppose that the spine-subtree $S_{\alpha-2}$ contains a leaf $\ell_{\alpha-2}$
which is $f$-dominated only by $\ell_{\alpha-3}$.
In that case, we consider the mapping $g_{\alpha-2}$ defined by 
$g_{\alpha-2}(\ell)=3$, $g_{\alpha-2}(\ell_j)=2$ for one leaf $\ell_j$
of each spine-subtree $S_j$, $1\le j\le \alpha-4$, 
$g_{\alpha-2}(\ell_{\alpha-3})=g_{\alpha-2}(\ell_{\alpha-2})=3$,
and $g(v)=f(v)$ for every other vertex $v$ of $L$
(see Figure~\ref{fig:bords-cas-4}(c), where $S_1,\dots,S_{\alpha-4}$ and $S_{\alpha-2}$ are supposed to be of type $\SSS_1$,
all other cases being similar).
The mapping $g_{\alpha-2}$ is clearly an independent
broadcast on $L$, with 
$$\cost(g_{\alpha-2})=\cost(f)-\alpha-\beta+3+2(\alpha-4)+3+3=\cost(f)+(\alpha-\beta)+1>\cost(f),$$
which contradicts the optimality of $f$. 

Therefore, each leaf of $S_{\alpha-2}$ 
is $f$-dominated by $\ell_{\alpha-3}$ and (at least) one other vertex $x$.
Moreover, we necessarily have 
$f(\ell_{\alpha-3})-1\leq f(x)\leq f(\ell_{\alpha-3})$ if $S_{\alpha-2}$ is of type $\SSS_1$, or
$f(x)=f(\ell_{\alpha-3})$ if $S_{\alpha-2}$ is of type $\SSS_2$.
Hence, $f(x)\geq f(\ell_{\alpha-3})-1=\beta-1\ge 5$,
which implies the uniqueness of~$x$,
since a 2-leaf of $S_{2\alpha-6}$ and a 1-leaf of $S_{2\alpha-5}$ are at distance~4 from each other,
and $v_{2\alpha-4}$ cannot be an $f$-broadcast vertex by Corollary~\ref{cor:lobster}.

\begin{figure}
\begin{center}
\begin{tikzpicture}
\node at (0,0.4) {$S_{0}$};
\node at (2.2,0.4) {$S_{\alpha-3}$};
%
\CTun{0}{0}  \CTun{2.2}{0}  
\CTunL{0}{-1}{$\alpha$}  \CTunL{2.2}{-1}{$\beta\le 2$}
\draw[thick,dotted] (0,0) to (2.2,0);
\draw[thick] (0,0) to (0.3,0);
\draw[thick] (1.9,0) to (2.2,0);
\node at (4.8,-1) {$\longrightarrow$};
%
\node at (6.2+0,0.4) {$S_{0}$};
\node at (6.2+2.2,0.4) {$S_{\alpha-3}$};
%
\CTun{6.2+0}{0}  \CTun{6.2+2.2}{0}  
\CTunL{6.2+0}{-1}{$\alpha-1$}  \CTunL{6.2+2.2}{-1}{3}
\draw[thick,dotted] (6.2+0,0) to (6.2+2.2,0);
\draw[thick] (6.2+0,0) to (6.2+0.3,0);
\draw[thick] (6.2+1.9,0) to (6.2+2.2,0);
\node at (4.8,-3) {(a) Case 4, $\beta\leq 2$};
\end{tikzpicture}
\vskip 0.5cm
\begin{tikzpicture}
\node at (-1.2,0.4) {$S_{0}$};
\node at (0,0.4) {$S_{1}$};
\node at (2.2,0.4) {$S_{\alpha-4}$};
\node at (3.4,0.4) {$S_{\alpha-3}$};
%
\CTun{-1.2}{0} \CTunL{0}{0}{0}  \CTunL{2.2}{0}{0}  \CTun{3.4}{0}  
\CTunL{-1.2}{-1}{$\alpha$}  \CTunL{3.4}{-1}{$3\le\beta<\alpha-1$}
\draw[thick,dotted] (-1.2,0) to (2.2,0);
\draw[thick] (-1.2,0) to (0.3,0);
\draw[thick] (1.9,0) to (3.4,0);
\node at (4.8,-1) {$\longrightarrow$};
%
\node at (7.4-1.2,0.4) {$S_{0}$};
\node at (7.4+0,0.4) {$S_{1}$};
\node at (7.4+2.2,0.4) {$S_{\alpha-4}$};
\node at (7.4+3.4,0.4) {$S_{\alpha-3}$};
%
\CTun{7.4-1.2}{0} \CTunL{7.4+0}{0}{2}  \CTunL{7.4+2.2}{0}{2}  \CTun{7.4+3.4}{0}
\CTunL{7.4-1.2}{-1}{$3$}  \CTunL{7.4+3.4}{-1}{$3$}
\draw[thick,dotted] (7.4-1.2,0) to (7.4+2.2,0);
\draw[thick] (7.4-1.2,0) to (7.4+0.3,0);
\draw[thick] (7.4+1.9,0) to (7.4+3.4,0);
\node at (4.8,-3) {(b) Case 4, $3\le \beta < \alpha-1$};
\end{tikzpicture}
\vskip 0.5cm
\begin{tikzpicture}
\node at (-1.2,0.4) {$S_{0}$};
\node at (0,0.4) {$S_{1}$};
\node at (2.2,0.4) {$S_{\alpha-3}$};
\node at (3.4,0.4) {$S_{\alpha-2}$};
%
\CTun{-1.2}{0} \CTunL{0}{0}{0}  \CTun{2.2}{0}  \CTunL{3.4}{0}{0}   
\CTunL{-1.2}{-1}{$\alpha$}   \CTunL{2.2}{-1}{$\alpha-1\le\beta\le\alpha$}     
\draw[thick,dotted] (-1.2,0) to (2.2,0);
\draw[thick] (-1.2,0) to (0.3,0);
\draw[thick] (1.9,0) to (3.4,0);
\node at (4.8,-1) {$\longrightarrow$};
%
\node at (7.4-1.2,0.4) {$S_{0}$};
\node at (7.4+0,0.4) {$S_{1}$};
\node at (7.4+2.2,0.4) {$S_{\alpha-3}$};
\node at (7.4+3.4,0.4) {$S_{\alpha-2}$};
%
\CTun{7.4-1.2}{0} \CTunL{7.4+0}{0}{2}  \CTun{7.4+2.2}{0}  \CTunL{7.4+3.4}{0}{3}   
\CTunL{7.4-1.2}{-1}{$3$}   \CTunL{7.4+2.2}{-1}{$3$}
\draw[thick,dotted] (7.4-1.2,0) to (7.4+2.2,0);
\draw[thick] (7.4-1.2,0) to (7.4+0.3,0);
\draw[thick] (7.4+1.9,0) to (7.4+3.4,0);
\node at (4.8,-3) {(c) Case 4, $6\le \alpha-1\le \beta \le \alpha$ and $\ell_{\alpha-2}$ is only $f$-dominated
by $\ell_{\alpha-3}$};
%
\end{tikzpicture}
\caption{\label{fig:bords-cas-4}Independent broadcasts for the proof of Lemma~\ref{lem:end-spine-subtrees}, case 4 (only one branch per spine-subtree is depicted).
}
\end{center}

\end{figure}
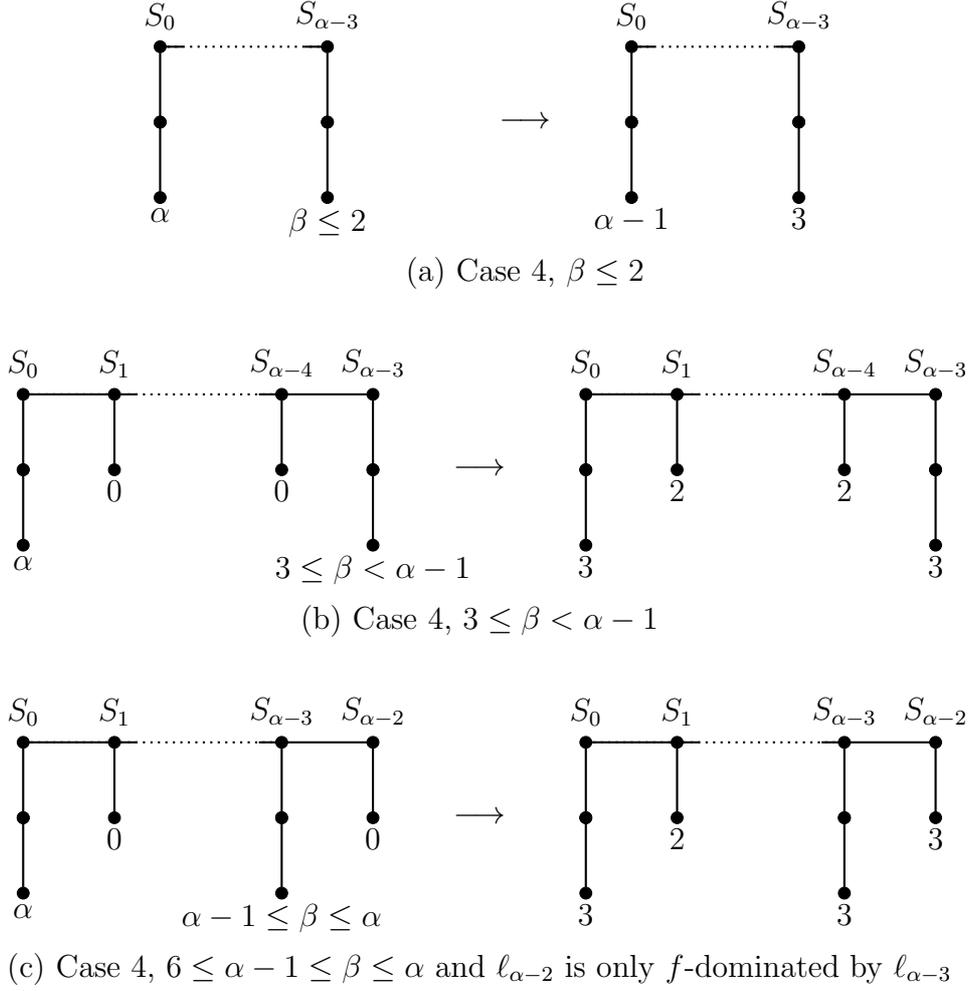

We consider two subcases, depending on whether $x$ is a 
1-leaf of $S_{2\alpha-5}$ or a 2-leaf of $S_{2\alpha-6}$.

\begin{enumerate}
\item {\it $x$ is a 1-leaf of $S_{2\alpha-5}$}.\\
In this case, we consider the mapping $g_{x}$ 
defined by $g_x(\ell)=g_x(\ell_{\alpha-3})=3$, $g_{x}(x)=2$,
$g_x(\ell_j)=2$ for one leaf $\ell_j$ for each of the spine-subtrees
$S_j$, $j\in\{1,\dots,2\alpha-6\}\setminus\{\alpha-3\}$,
and $g_{x}(v)=f(v)$ for every other vertex $v$ of $L$
(see Figure~\ref{fig:bords-cas-4-bis}(a)).
Again, the mapping $g_{x}$ is clearly an independent
broadcast on $L$, with 
$$\begin{array}{rcl}
\cost(g_{x}) & = & \cost(f)-\alpha+3 + 2(\alpha-4) -\beta+3 + 2(\alpha-3)-f(x)+2\\
 & = & \cost(f) + 3\alpha - \beta - f(x) - 6.
\end{array}$$
We thus get a contradiction on the minimality of $\alpha$ or the optimality of $f$ 
since $\alpha\ge 7$ and $f(x)\leq\beta\leq\alpha$.

\item {\it $x$ is a 2-leaf of $S_{2\alpha-6}$}.\\
In this case, we consider the mapping $g_{x}$ 
defined by $g_x(\ell)=g_x(\ell_{\alpha-3})=g_{x}(x)=3$,
$g_x(\ell_j)=2$ for one leaf $\ell_j$ for each of the spine-subtrees
$S_j$, $j\in\{1,\dots,2\alpha-7\}\setminus\{\alpha-3\}$,
and $g_{x}(v)=f(v)$ for every other vertex $v$ of $L$
(see Figure~\ref{fig:bords-cas-4-bis}(b)).
Again, the mapping $g_{x}$ is clearly an independent
broadcast on $L$, with 
$$\begin{array}{rcl}
\cost(g_{x}) & = & \cost(f)-\alpha+3 + 2(\alpha-4) -\beta+3 + 2(\alpha-4)-f(x)+3\\
 & = & \cost(f) + 3\alpha - \beta - f(x) - 7.
\end{array}$$
We thus get a contradiction on the minimality of $\alpha$ or the optimality of $f$ 
since $\alpha\ge 7$ and $f(x)\leq\beta\leq\alpha$.

\end{enumerate}
\end{enumerate}

We thus obtain a contradiction in each case,
which implies that no counter-example to the lemma exists.
This completes the proof.
\end{proof}

\begin{figure}
\begin{center}
\begin{tikzpicture}
\node at (-1.2,0.4) {$S_{0}$};
\node at (0,0.4) {$S_{1}$};
\node at (1.8,0.4) {$S_{\alpha-3}$};
\node at (3,0.4) {$S_{\alpha-2}$};
\node at (4.8,0.4) {$S_{2\alpha-5}$};
\CTun{-1.2}{0} \CTunL{0}{0}{0}  \CTun{1.8}{0}  \CTunL{3}{0}{0}  \CTunL{4.8}{0}{$f(x)$}
\CTunL{-1.2}{-1}{$\alpha$}  \CTunL{1.8}{-1}{$\beta$}
\draw[thick,dotted] (-1.2,0) to (1.8,0);
\draw[thick,dotted] (3,0) to (4.8,0);
\draw[thick] (-1.2,0) to (0.3,0);
\draw[thick] (1.5,0) to (3.3,0);
\draw[thick] (4.5,0) to (4.8,0);
\node at (6.2,-1) {$\longrightarrow$};
\node at (8.8-1.2,0.4) {$S_{0}$};
\node at (8.8+0,0.4) {$S_{1}$};
\node at (8.8+1.8,0.4) {$S_{\alpha-3}$};
\node at (8.8+3,0.4) {$S_{\alpha-2}$};
\node at (8.8+4.8,0.4) {$S_{2\alpha-5}$};
\CTun{8.8-1.2}{0} \CTunL{8.8+0}{0}{2}  \CTun{8.8+1.8}{0}  \CTunL{8.8+3}{0}{2}  \CTunL{8.8+4.8}{0}{2}    
\CTunL{8.8-1.2}{-1}{$3$}  \CTunL{8.8+1.8}{-1}{$3$}
\draw[thick,dotted] (8.8-1.2,0) to (8.8+1.8,0);
\draw[thick,dotted] (8.8+3,0) to (8.8+4.8,0);
\draw[thick] (8.8-1.2,0) to (8.8+0.3,0);
\draw[thick] (8.8+1.5,0) to (8.8+3.3,0);
\draw[thick] (8.8+4.5,0) to (8.8+4.8,0);
\node at (6.2,-3) {(a) Case 4(a), $\beta\ge 7$, $\beta-1\le f(x)\le\beta$ 
and $\ell_{\alpha-2}$ is $f$-dominated by $\ell_{\alpha-3}$ and $x=\ell_{2\alpha-5}$};
\end{tikzpicture}
\vskip 0.5cm
\begin{tikzpicture}
\node at (-1.2,0.4) {$S_{0}$};
\node at (0,0.4) {$S_{1}$};
\node at (1.8,0.4) {$S_{\alpha-3}$};
\node at (3,0.4) {$S_{\alpha-2}$};
\node at (4.8,0.4) {$S_{2\alpha-6}$};
\CTun{-1.2}{0} \CTunL{0}{0}{0}  \CTun{1.8}{0}  \CTunL{3}{0}{0}  \CTun{4.8}{0}
\CTunL{-1.2}{-1}{$\alpha$}  \CTunL{1.8}{-1}{$\beta$}   \CTunL{4.8}{-1}{$f(x)$}
\draw[thick,dotted] (-1.2,0) to (1.8,0);
\draw[thick,dotted] (3,0) to (4.8,0);
\draw[thick] (-1.2,0) to (0.3,0);
\draw[thick] (1.5,0) to (3.3,0);
\draw[thick] (4.5,0) to (4.8,0);
\node at (6.2,-1) {$\longrightarrow$};
\node at (8.8-1.2,0.4) {$S_{0}$};
\node at (8.8+0,0.4) {$S_{1}$};
\node at (8.8+1.8,0.4) {$S_{\alpha-3}$};
\node at (8.8+3,0.4) {$S_{\alpha-2}$};
\node at (8.8+4.8,0.4) {$S_{2\alpha-6}$};
\CTun{8.8-1.2}{0} \CTunL{8.8+0}{0}{2}  \CTun{8.8+1.8}{0}  \CTunL{8.8+3}{0}{2}  \CTun{8.8+4.8}{0}  
\CTunL{8.8-1.2}{-1}{$3$}  \CTunL{8.8+1.8}{-1}{$3$}   \CTunL{8.8+4.8}{-1}{$3$}
\draw[thick,dotted] (8.8-1.2,0) to (8.8+1.8,0);
\draw[thick,dotted] (8.8+3,0) to (8.8+4.8,0);
\draw[thick] (8.8-1.2,0) to (8.8+0.3,0);
\draw[thick] (8.8+1.5,0) to (8.8+3.3,0);
\draw[thick] (8.8+4.5,0) to (8.8+4.8,0);
\node at (6.2,-3) {(b) Case 4(b), $\beta\ge 7$, $\beta-1\le f(x)\le\beta$ 
and $\ell_{\alpha-2}$ is $f$-dominated by $\ell_{\alpha-3}$ and $x=\ell_{2\alpha-6}$};
\end{tikzpicture}
\caption{\label{fig:bords-cas-4-bis}Independent broadcasts for the proof of Lemma~\ref{lem:end-spine-subtrees}, case 4 continued (only one branch per spine-subtree is depicted).
}
\end{center}

\end{figure}


Let $f$ be any independent broadcast on a locally uniform $2$-lobster $L$ and 
$S_i$ be any spine-subtree  of~$L$.
Recall that $f^*(S_i)$ denotes the broadcast value of~$S_i$,
that is, the sum of the broadcast values of the vertices of $S_i$.
The next lemma shows that if $f$ is an optimal independent broadcast
on a locally uniform $2$-lobster $L$ of length $k\ge 1$, 
then there exists an optimal independent broadcast $\tf$ on $L$
such that 
every spine-subtree of $L$ contains an $\tf$-broadcast vertex.

\begin{lemma}
If $L$ is a locally uniform $2$-lobster of length $k\ge 1$, 
and $f$ is an optimal independent broadcast on $L$, 
then 
there exists an optimal independent broadcast $\tf$ on~$L$
such that
$\tf^*(S_i)>0$ for every $i$, $0\le i\le k$.
\label{lem:no-null-spine-subtree}
\end{lemma}

\begin{proof}
Since the result directly follows from Lemma~\ref{lem:broadcast-vertex-in-S0-Sk} if $k=1$,
we only need to consider the case $k>1$.
Assume to the contrary that there does not exist any such independent broadcast $\tf$,
and let $f$ be an optimal independent broadcast on $L$
that maximises the number of $f$-broadcast leaves.
Let~$i$ be the smallest index such that $f^*(S_i)=0$, and $\ell_i$ be any leaf of $S_i$.

Suppose first that $S_i$ is of type $\SSS_2$, so that $\ell_i$ is a 2-leaf.
From the choice of $S_i$, we know that $f^*(S_{i-1})>0$.
Moreover, we claim that 
$S_{i-1}$ is of type $\SSS_2$ and that,
for every leaf $\ell_{i-1}$ of $S_{i-1}$,
$f(\ell_{i-1})\leq 4$.
Indeed, this follows from Lemma~\ref{lem:end-spine-subtrees} if $i=1$.
If $i\geq 2$, $S_{i-1}$ cannot be of type $\SSS_1$ since otherwise
any vertex that $f$-dominates $\ell_i$ would also $f$-dominate all
the leaves of $S_{i-1}$ (such a dominating vertex must belong
to some $S_j$ with $j>i$). 
Thus $S_{i-1}$ is of type $\SSS_2$
and the fact that $f^*(S_{i-2})>0$ implies $f(\ell_{i-1})\leq 4$
for every leaf $\ell_{i-1}$ of $S_{i-1}$.
Therefore, $\ell_i$ is necessarily $f$-dominated by a
vertex $y\in S_j$, for some $j>i$.
Consider the mapping~$g$ defined by $g(y)=f(y)-1$,
$g(\ell_i)=1$, and $g(v)=f(v)$ for every other vertex $v$ of~$L$.
The mapping~$g$ is
an independent broadcast on $L$ with $\cost(g)=\cost(f)$,
which contradicts the maximality of the number of $f$-broadcast leaves.

Suppose now that $S_i$ is of type $\SSS_1$, so that $\ell_i$ is a 1-leaf,
and that $\ell_i$ is $f$-dominated by a unique vertex $y$.
If $y\in S_{i-1}$ then the mapping~$g$ defined by $g(y)=f(y)-1$,
$g(\ell_i)=1$, and $g(v)=f(v)$ for every other vertex $v$ of~$L$, is
an independent broadcast on $L$, with $\cost(g)=\cost(f)$,
which contradicts the maximality of the number of $f$-broadcast leaves.
If $y\in S_j$ for some $j>i$, then we necessarily have either
$f(y)=d_L(y,\ell_i)$, if $S_{i-1}$ is of type $\SSS_1$, or
$d_L(y,\ell_i)\leq f(y)\leq d_L(y,\ell_i)+1$, if $S_{i-1}$ is of type $\SSS_2$.
Therefore,
the mapping~$g$ defined by $g(y)=d_L(y,\ell_i)-1$,
$g(\ell_i)=1+f(y)-d_L(y,\ell_i)$, and $g(v)=f(v)$ for every other vertex $v$ of~$L$, is
an independent broadcast on $L$, with $\cost(g)=\cost(f)$,
which again contradicts the maximality of the number of $f$-broadcast leaves.

Suppose finally that $S_i$ is of type $\SSS_1$
and that $\ell_i$ is $f$-dominated by 
two distinct vertices $y_1$ and $y_2$,
with $y_1\in S_{i_1}$ and $y_2\in S_{i_2}$. Note that we necessarily have, without loss
of generality, $i_1=i-1$ and $i<i_2$.
We claim that $f(y_1)\ge 3$ and $f(y_2)\ge 3$.
Indeed, if, say, $f(y_1)=2$, then $y_1=v_{i-1}$,
which contradicts Corollary~\ref{cor:lobster}.
The case $f(y_2)=2$ is similar.
Moreover, 
we clearly have either $f(y_1)=3$ and $f(y_2)=d_L(y_2,\ell_i)$, if $S_{i-1}$ is of type $\SSS_1$, or
$f(y_1)=4$ and 
$d_L(y_2,\ell_i)\leq f(y_2)\leq d_L(y_2,\ell_i)+1$, if $S_{i-1}$ is of type $\SSS_2$.
We consider two cases, depending on the value of $f(y_2)$.

\begin{enumerate}
\item {\em $f(y_2)=d_L(y_2,\ell_i)$}.\\
In that case, the mapping~$g$ defined by $g(y_1)=f(y_1)-1$, $g(y_2)=f(y_2)-1$,
$g(\ell_i)=2$, and $g(v)=f(v)$ for every other vertex $v$ of~$L$, is
an independent broadcast on $L$, with $\cost(g)=\cost(f)$, 
which contradicts the maximality of the number of $f$-broadcast leaves. 

\item {\em $f(y_2)=d_L(y_2,\ell_i)+1$}.\\
In that case, we necessarily have that
$S_{i-1}$ is of type $\SSS_2$ and $f(y_1)=4$ on one hand, 
and $f(y_2)\geq 4$ on the other hand,
which implies that $f^*(S_{i_2+1})=0$, if $i_2 < k$.

If $y_2$ is not a 1-leaf of $S_{i+1}$,
then the mapping~$g$ defined by $g(y_1)=3$, $g(y_2)=d_L(y_2,\ell_i)-1$,
$g(\ell_i)=2+f(y_2)-d_L(y_2,\ell_i)$, and $g(v)=f(v)$ for every other vertex $v$ of~$L$, is
an independent broadcast on $L$, with $\cost(g)=\cost(f)$,
which contradicts the maximality of the number of $f$-broadcast leaves.

Otherwise, that is, $y_2$ is a 1-leaf of $S_{i+1}$ and $f(y_2)=4$,
we cannot give to $\ell_i$ the broadcast value $2+f(y_2)-d_L(y_2,\ell_i)=2+4-3=3$,
since $\ell_i$ would then dominate $y_2$.
Observe that since $S_{i+1}$ is of type $\SSS_1$, we have $i+1<k$, so that $S_{i+2}$ exists.

If $S_{i+2}$ is of type $\SSS_2$, then the leaves of $S_{i+2}$ are necessarily
$f$-dominated only by $y_2$. Let $\ell_{i+2}$ be any leaf of $S_{i+2}$.
In that case, the mapping $g$  defined by $g(y_1)=3$, $g(\ell_i)=2$, $g(y_2)=2$,
$g(\ell_{i+2})=1$, and $g(v)=f(v)$ for every other vertex $v$ of~$L$, is
an independent broadcast on $L$, with $\cost(g)=\cost(f)$,
which contradicts the maximality of the number of $f$-broadcast leaves.

Suppose finally that $S_{i+2}$ is of type $\SSS_1$, and let $\ell_{i+2}$
denote any 1-leaf of $S_{i+2}$. Note that $y_2$ $f$-dominates $\ell_{i+2}$.
If $\ell_{i+2}$ is $f$-dominated only by $y_2$,
then the mapping $g$  defined by $g(y_1)=3$, $g(\ell_i)=2$, $g(y_2)=2$,
$g(\ell_{i+2})=2$, and $g(v)=f(v)$ for every other vertex $v$ of~$L$, is
an independent broadcast on $L$, with $\cost(g) > \cost(f)$,
which contradicts the optimality of $f$.
Otherwise, let $z$ be the other vertex of $L$ which $f$-dominates $\ell_{i+2}$.
Then,
the mapping $g$  defined by $g(y_1)=3$, $g(\ell_i)=2$, $g(y_2)=2$,
$g(\ell_{i+2})=2$, $g(z)=g(z)-1$, and $g(v)=f(v)$ for every other vertex $v$ of~$L$, is
an independent broadcast on $L$, with $\cost(g)=\cost(f)$,
which contradicts the maximality of the number of $f$-broadcast leaves.
\end{enumerate}

This completes the proof.
\end{proof}

From Lemma~\ref{lem:no-null-spine-subtree} and Corollary~\ref{cor:lobster}, we get the following corollary.

\begin{corollary}
If $L$ is a locally uniform $2$-lobster of length $k\ge 1$, 
and $f$ is an optimal independent broadcast on $L$, 
then 
there exists an optimal independent broadcast $\tf$ on~$L$
such that,
for every spine-subtree $S_i$ of $L$, $0\le i\le k$, and every vertex $x$ of $S_i$,
$\tf(x)\le 1$ if $x=v_i$,
$\tf(x)\le 3$ if $x$ is a 1-leaf of $S_i$,
and $\tf(x)\le 4$ if $x$ is a 2-leaf of $S_i$.
\label{cor:no-null-spine-subtree}
\end{corollary}

\begin{proof}
Let $\tf$ be an optimal independent broadcast on $L$ such that
$\tf^*(S_i)>0$ for every~$i$, $0\le i\le k$.
The existence of~$\tf$ is guaranteed by Lemma~\ref{lem:no-null-spine-subtree}.
If $x=v_i$, then $\tf(x)\le 1$ follows from Corollary~\ref{cor:lobster}.
Otherwise, assuming that the claimed bound on~$\tf(x)$ is not satisfied would imply
$\tf^*(S)=0$ for a neighbouring spine-subtree $S$ of~$S_i$, in contradiction with
our assumption on~$\tf$.
\end{proof}

The next two lemmas show that if $f$ is an optimal independent broadcast
on a locally uniform $2$-lobster $L$ of length $k\ge 1$, 
then there exists an optimal independent broadcast $\tf$ on $L$
such that the $\tf$-value of every spine-subtree $S$ of $L$
is bounded from above by a value depending on the type of $S$.

Recall that $\TTT_4$ denotes the type of sequence used 
in step 4 in the proof of Lemma~\ref{lem:beta-star-general},
that is
$$\TTT_4 = \langle\overline{\EEE_2}.\EEE_0\rangle\ \FFFA.\{(\XXX|\FFFA).\FFFA\}^*\ \langle\EEE_0.\overline{\EEE_2}\rangle.$$

In the following, when we say that a spine-subtree $S_i$ {\it appears as an $\FFFA$-spine-subtree}
(resp. {\it as an $\XXX$-spine-subtree}) {\it in a sequence of type $\TTT_4$}, 
we mean that $S_i=A_j$ (resp. $S_i=X_j$) for some $j$, $0\le j\le p$ (resp. $1\le j\le p$), 
in the corresponding sequence $A_0X_1A_1\dots X_pA_p$.

\begin{lemma}
If $L$ is a locally uniform $2$-lobster of length $k\ge 1$, 
and $f$ is an optimal independent broadcast on $L$, 
then 
there exists an optimal independent broadcast $\tf$ on~$L$
such that, for every spine-subtree $S_i$ of $L$, $0\le i\le k$, $\tf$ satisfies the following properties.
\begin{enumerate}
\item \label{item:lem} $\tf^*(S_i)>0$.
\item \label{item:1-partout} If $\tf^*(S_i)=\lambda_1(S_i)$, or $\tf^*(S_i)=\lambda_2(S_i)$,
then $\tf(\ell)=1$ for every leaf of $S_i$.
\item If $S_i$ is of type $\SSS_1$, then
  \begin{enumerate}
  \item \label{item:GXa} $\tf^*(S_i)\leq \lambda_1(S_i)$ if $S_i$ is of type $\GGG$ or $\XXX_a$,
  \item \label{item:Xc} $\tf^*(S_i)\leq 3$ if $S_i$ is of type $\XXX_c$,
  \item \label{item:XcSEQ} $\tf^*(S_i)\leq 2$ if $S_i$ is of type $\XXX_c$ 
  and $S_i$ belongs to a sequence of type $\TTT_4$, 
  \end{enumerate}
\item If $S_i$ is of type $\SSS_2$, then 
  \begin{enumerate}
  \item \label{item:FaSEQ} $\tf^*(S_i)\leq \lambda_2(S_i) + 1$ if $S_i$ is of type
      $\FFFA$, 
  \item \label{item:XbYc} $\tf^*(S_i)\leq \lambda_2(S_i)+\lambda_2^*(S_i)+\alpha_2^*(S_i)$ if $S_i$ is of type $\XXX_b$ 
      or $\YYYC$,
  \item \label{item:XbSEQ} $\tf^*(S_i)\leq \lambda_2(S_i)+\lambda_2^*(S_i)+\alpha_2^*(S_i) - 1$ if $S_i$ is of type $\XXX_b$
  and $S_i$ belongs to a sequence of type $\TTT_4$.
  \end{enumerate}
\item \label{item:racinePas1}
  If $S_i$ is not of type $\FFFA$, then $\tf(v_i)=0$.
\end{enumerate}
\label{lem:not-more-than-beta-star}
\end{lemma}

\begin{proof}
Thanks to Lemma~\ref{lem:no-null-spine-subtree}, we know that we can choose
an independent broadcast $\tf$ on $L$ which satisfies Item~\ref{item:lem}.
By Corollary~\ref{cor:no-null-spine-subtree}, we get
that the $\tf$-value of every 1-leaf is at most~3,
and that the $\tf$-value of every 2-leaf is at most~4.
This observation will be implicitly used all along the proof.

Note also that if $\tf^*(S_i)=\lambda_1(S_i)$, or $\tf^*(S_i)=\lambda_2(S_i)$,
then we can obviously modify $\tf$, in order to satisfy Item~\ref{item:1-partout},
without modifying its cost.

We now prove that $\tf$ can be chosen in such a way that
it satisfies all the other items of the lemma.
Let $S_i$ be any spine-subtree of $L$.

Note first that the above observation already proves Item~\ref{item:Xc}.
Moreover, observe that
we necessarily have
$\tfs(S_i)\le \lambda_1(S_i)$ if $S_i$ is of type $\GGG$ or $\XXX_a$,
since in each of these cases, the only way to attain this value
is to have one leaf $\ell$ of $S_i$ with $\tf(\ell)=\tfs(S_i)$,
which would imply that a neighbouring spine-subtree of $S_i$
has no $\tf$-broadcast vertex, in contradiction with Item~\ref{item:lem}.
This proves Item~\ref{item:GXa}.

Suppose now that $S_i$ is of type $\FFFA$.
If the broadcast value of a 2-leaf of $S_i$ is $2$,
then its at least two sister-leaves cannot
be $\tf$-broadcast vertices since this would contradict the optimality of $\tf$
(by giving the broadcast value $1$ to each 2-leaf of a branch $B$ of $S_i$, we get
$\tfs(B)\geq\lambda_2(B)$).
Therefore, the greatest possible value of $\tfs(S_i)$ is obtained when
the spine-vertex $v_i$ and all the 2-leaves of $S_i$ have $\tf$-value $1$.
This gives $\tfs(S_i)\leq \lambda_2(S_i)+1$, 
which proves Item~\ref{item:FaSEQ}.

Suppose now that $S_i$ is of type $\XXX_b$ or $\YYYC$.
Observe first that, if $\tf(v_i)=1$, then 
the broadcast value of every leaf of $S_i$ is $1$.
The optimality of $\tf$ then implies that $S_i$ has a unique branch
$B$ with two leaves and no branch with a unique leaf, since otherwise we could set to~$0$ the broadcast
value of $v_i$ and to~$3$ the broadcast value of one leaf of every such branch,
and thus increase the cost of $\tf$.
(Note also that, for the same cost, we can set $\tf(v_i)=0$ and set to 0 and~3 the broadcast
value of the two leaves of this branch. This remark will be useful in the next paragraph.)
In that case, we thus have 
$\tf^*(S_i) = \lambda_2(S_i)+1 \leq \lambda_2(S_i)+\lambda_2^*(S_i)+\alpha_2^*(S_i)$,
which proves Item~\ref{item:XbYc}.
Suppose now $\tf(v_i)=0$ and let $B$ be any branch of $S_i$. 
The optimality of $\tf$ then implies the following.
If $B$ has one or two 2-leaves, the $\tf$-value of one of theses leaves is $3$
(otherwise, we would have $\tfs(B)\le 2$).
If $B$ has at least three leaves, the largest possible value of $\tfs(B)$ is $\lambda_2(B)$,
since as soon as a 2-leaf has a broadcast value at least~2, none of its sister-leaves can be a broadcast
vertex. (Note that if $B$ has three 2-leaves, then either one of them has $\tf$-value~$3$,
or, for the same cost, each of them has $\tf$-value~$1$.)
Therefore, $\tf^*(S_i)\leq \lambda_2(S_i)+\lambda_2^*(S_i)+\alpha_2^*(S_i)$,
which proves Item~\ref{item:XbYc}.

\medskip

Suppose now that $S_i$ is of type $\XXX_b$ or $\XXX_c$, and belongs
to some sequence of type $\TTT_4$. 
In such a sequence, each spine-subtree of type $\XXX$
is associated with one of its neighbouring spine-subtrees of type $\FFFA$, in such
a way that no spine-subtree of type $\FFFA$ is associated with two distinct
spine-subtrees of type $\XXX$.
Let $S'_i$ denote the spine-subtree associated with $S_i$ (we have $S'_i\in\{S_{i-1},S_{i+1}\}$).
On the one hand, from the above discussion about spine-subtrees of type $\FFFA$, we know that their
largest possible  broadcast value can be attained only if their spine-vertex has
broadcast value~$1$.
On the other hand, from the above discussion about spine-subtrees of type 
$\XXX_b$ or $\XXX_c$, we know that their
largest possible broadcast value can be attained only if one leaf $\ell$ of 
the unique branch $B$ of $S_i$ having at most two, or exactly two leaves, 
has a broadcast value of $3$.
Therefore, $S_i$ and $S'_i$ cannot get their largest possible broadcast value
both at the same time. We thus need either to remove the broadcast value of the spine-vertex
of $S'_i$, or to give the broadcast value
$2$ to $\ell$ if $\ell$ is a 2-only-leaf,
or $1$ to each 2-leaf of $B$ otherwise.
This second choice proves that the optimal independent broadcast $\tf$ can be chosen
in order to satisfy Items \ref{item:XbSEQ} and~\ref{item:XcSEQ}.

\medskip

It remains to prove Item~\ref{item:racinePas1}.
If $S_i$ is of type $\SSS_1$, the result follows from Lemma~\ref{lem:branch-0}.
We thus only need to consider the case when $S_i$ is of type $\XXX_b$ or~$\YYYC$.
Suppose that $\tf(v_i)=1$ (we cannot have $\tf(v_i)>1$ by Corollary~\ref{cor:lobster}).
If $S_i$ is of type~$\XXX_b$, then we can set $\tf(v_i)=0$,
$\tf(\ell)=3$ for a 2-leaf $\ell$ of~$S_i$ belonging to the unique branch of $S_i$
having at most two leaves, and $\tf(\ell')=0$ for the sister-leaf $\ell'$ of
$\ell$, if any. Such a modification does not decrease the cost of $\tf$ and we are done.
If $S_i$ is of type~$\YYYC$, then we cannot have $\tf(v_i)=1$,
since this would contradict the optimality of $\tf$,
as the previous modification can be done on the at least two branches of $S_i$
having at most two 2-leaves.

This completes the proof.
\end{proof}

\begin{lemma}
If $L$ is a locally uniform $2$-lobster of length $k\ge 1$, 
and $f$ is an optimal independent broadcast on $L$, 
then 
there exists an optimal independent broadcast $\tf$ on~$L$
such that, for every spine-subtree $S_i$ of $L$, $0\le i\le k$, $\tf$ satisfies the following properties.
\begin{enumerate}
\item \label{item:previous-lem} $\tf$ satisfies all the items of Lemma~\ref{lem:not-more-than-beta-star},
\item \label{item:FaFbPASSEQ} $\tf^*(S_i)\leq \lambda_2(S_i)$ if $S_i$ is of type
$\FFFA$ and $S_i$ does 
not appear as an $\FFFA$-spine-subtree in a sequence of type $\TTT_4$.    
\end{enumerate}
\label{lem:not-more-than-beta-star-total}
\end{lemma}

\begin{proof}
Thanks to Lemma~\ref{lem:not-more-than-beta-star}, we know that we can choose
an independent broadcast on~$L$ which satisfies Item~\ref{item:previous-lem}.
So consider such a broadcast $\tf$ on~$L$.
Recall that by Item~\ref{item:racinePas1} of Lemma~\ref{lem:not-more-than-beta-star},
$\tf(v_i)=0$ for every spine-subtree $S_i$ which is not of type $\FFFA$. 
Moreover, by Item~\ref{item:FaSEQ} of Lemma~\ref{lem:not-more-than-beta-star},
if $S_i$ is a spine-subtree of $L$ of type $\FFFA$,
then $\tfs(S_i)\leq \lambda_2(S_i)+1$, and, as observed in the proof of that lemma, the only way
to attain this value is to give a broadcast value of~$1$ to the spine-vertex
and to all the 2-leaves of $S_i$.

Suppose that there exists in $L$ a spine-subtree $S_i$ of type $\FFFA$,
that does not appear as an $\FFFA$-spine-subtree in any sequence of type $\TTT_4$,
and such that $\tf(v_i)=1$, which implies $\tf(\ell) = 1$ for every leaf $\ell$ of~$S_i$
since $\tf$ is optimal. 
Such a spine-subtree will be called a {\it bad spine-subtree}.
Moreover, suppose that $S_i$ is the leftmost such bad spine-subtree of~$L$.
We claim that the broadcast $\tf$ can be modified, without decreasing its cost,
in such a way that either 
the number of bad spine-subtrees in $L$ strictly decreases,
or this number is still the same but the index of the leftmost bad spine-subtree in $L$ strictly increases,
which will prove Item~\ref{item:FaFbPASSEQ}.

All along this proof, we will modify the independent broadcast $\tf$ on some
spine-subtrees, according to their type.
By applying the \emph{standard modification of $\tf$} on a spine-subtree $S_j$ of $L$,
we mean the following.
\begin{itemize}
\item If $S_j$ is a bad spine-subtree, then we set $\tf(v_j)=0$.
\item If $S_j$ is of type $\FFFA$ and is not a bad spine-subtree, then we set $\tf(v_j)=1$.
\item If $S_j$ is of type $\XXX_b$, then we set $\tf(\ell)=3$ for a 2-leaf $\ell$ of the unique branch of $S_j$
having at most two leaves.
\item If $S_j$ is of type $\XXX_c$, then we set $\tf(\ell)=3$ for a 1-leaf $\ell$ of $S_j$,
and $\tf(\ell')=0$ for the sister-leaf of~$\ell$.
\item If $S_j$ is of type $\YYYC$, then we set $\tf(\ell)=3$ for one 2-leaf~$\ell$ of each branch of~$S_j$ 
having at most two 2-leaves, and $\tf(\ell')=0$ for its sister-leaf $\ell'$, if any.
\end{itemize}

We first consider the case when $S_i$ belongs to some sequence of type $\TTT_4$,
but not as an $\FFFA$-spine-subtree, which implies that the length of this sequence is at least~3.
Let $A_0X_1A_1\dots X_pA_p$, $p\ge 1$, denote the corresponding sequence.
Every spine-subtree corresponding to some $X_j$, $1\le j\le p$, is surrounded
by two spine-subtrees of type $\FFFA$.
If any such spine-subtree $S_{\alpha}$ (corresponding to some $X_j$) is of type $\FFFA$, then
having $\tf(s_{\alpha})=1$ would imply that none of $s_{\alpha-1}$ and $s_{\alpha+1}$ is an $\tf$-broadcast
vertex.
We can thus apply the standard modification of $\tf$ on all the spine-subtrees
$A_0,X_1,A_1,\dots, X_p,A_p$,
without decreasing the cost of $\tf$.

\medskip

We suppose now that $S_i$ does not belong to
any sequence of type $\TTT_4$. The following claims will be useful in the sequel.

\begin{claim}
If a bad spine-subtree $S_j$, not belonging to any sequence of type $\TTT_4$, 
has a neighbouring spine-subtree of type~$\YYYC$,
then we can modify~$\tf$, without decreasing its cost, in such
a way that the number of bad spine-subtrees strictly decreases.
\label{cl:pas-Fa-Yc}
\end{claim}

\begin{proof}
Suppose that $S_{i+1}$ exists and is of type $\YYYC$ 
(the case when $S_{i-1}$ exists and is of type $\YYYC$ is similar).
If $S_{i+2}$ does not exist, or if $S_{i+2}$ exists and $\tf(v_{i+2})=0$,
then we can apply the standard modification 
of $\tf$ on $S_i$ and $S_{i+1}$, without decreasing the cost of $\tf$.
Finally, if $S_{i+2}$ exists and $\tf(v_{i+2})=1$, 
we get that $S_{i+2}$ is also a bad spine-subtree of $L$.
In that case, we can apply the standard modification of $\tf$ on $S_i$, $S_{i+1}$
and~$S_{i+2}$, without decreasing the cost of $\tf$ since $S_{i+1}$ has at least 
two branches with at most two leaves.
\end{proof}

\begin{claim}
Let $S_i$ be the leftmost bad spine-subtree in $L$, that does not belong to any sequence of type~$\TTT_4$.
If $S_{i+1}$ is of type $\FFFA$ and $S_{i+2}$ is of type $\GGG$, $\XXX_a$ or~$\XXX_c$,
then we can modify~$\tf$, without decreasing its cost, in such
a way that the index of the leftmost bad spine-subtree in $L$ strictly increases.
\label{cl:pas-FaFa-GXaXc}
\end{claim}

\begin{proof}
Since $\tf(v_i)=1$, we get that $\tf^*(S_{i+2})=\lambda_1(S_{i+2})$
(the $\tf$-broadcast value of every leaf in $S_{i+2}$ is~1), so that we can apply
the standard modification of $\tf$ on $S_i$ and $S_{i+1}$
(recall that $\tf(\ell)=1$ for every leaf of $S_{i+2}$ by Item~\ref{item:1-partout}
of Lemma~\ref{lem:not-more-than-beta-star}), without decreasing the cost of $\tf$,
but strictly increasing the index of the leftmost bad spine-subtree in $L$.
\end{proof}

\begin{claim}
If $S_i$ is the leftmost bad spine-subtree in $L$, that does not belong to any sequence of type~$\TTT_4$,
then we can assume that $S_{i+1}$ is of type $\XXX_b$ or $\XXX_c$.
\label{cl:droite}
\end{claim}

\begin{proof}
If $S_{i-1}$ is of type $\XXX_c$, then,
since $\tf(v_i)=1$, we get that $\tf(\ell)=1$ for every
1-leaf of $S_{i-1}$.
Similarly, if $S_{i-1}$ is of type $\XXX_b$, then $\tf(\ell)\le 2$ for every
2-leaf of $S_{i-1}$.
In both cases, we can thus apply the standard modification 
of $\tf$ on $S_i$ and $S_{i-1}$, without decreasing the cost of $\tf$.

Thanks to Claim~\ref{cl:pas-Fa-Yc}, we can thus assume that $S_{i-1}$ either does not exist 
or is of type $\FFFA$, $\XXX_a$
or~$\GGG$ (these two latter cases imply that $S_{i-2}$ cannot be of type~$\XXX_c$,
and is thus necessarily of type $\GGG$ or~$\XXX_a$).
If $S_{i-1}$ is of type $\FFFA$ and $S_{i-2}$ is of type $\XXX_c$,
then, similarly as above, we can apply the standard modification 
of $\tf$ on $S_i$ and $S_{i-2}$, without decreasing the cost of $\tf$.

In all the remaining cases, $S_i$ is of type $\langle\overline{\EEE_2}.\EEE_0\rangle\FFFA$,
so that we necessarily have that 
either $S_{i+1}$ is of type $\XXX_b$ or $\XXX_c$,
or $S_{i+1}$ is of type $\FFFA$ and $S_{i+2}$ is of type $\XXX_c$, since otherwise
$S_i$ would be a sequence of type $\TTT_4$.
But the case when $S_{i+1}$ is of type $\FFFA$ and $S_{i+2}$ is of type $\XXX_c$
is covered by Claim~\ref{cl:pas-FaFa-GXaXc}.

This concludes the proof.
\end{proof}

Thanks to Claims \ref{cl:pas-Fa-Yc} and~\ref{cl:droite}, we can assume that 
$S_{i-1}$ is not of type~$\YYYC$, and that $S_{i+1}$ is of type either $\XXX_b$ or~$\XXX_c$.
We consider these two cases separately when $S_{i+2}$ is not a bad spine-subtree,
or together otherwise.

\begin{enumerate}
\item $S_{i+1}$ is of type $\XXX_b$ and $S_{i+2}$ is not a bad spine-subtree.\\
In that case, we can apply the standard modification of $\tf$ on $S_i$ and
$S_{i+1}$ and we are done.

\item $S_{i+1}$ is of type $\XXX_c$ and $S_{i+2}$ is not a bad spine-subtree.\\
In this case, $S_{i+2}$ necessarily exists and is of type~$S_2$.
If $S_{i+3}$ is not a bad spine-subtree, then 
we can apply the standard modification of $\tf$ on $S_i$ and
$S_{i+1}$ and we are done.
Otherwise, $S_{i+2}$ must be of type $\XXX_b$ or $\FFFA$, thanks to our assumption based 
on Claim~\ref{cl:pas-Fa-Yc}.
We thus have two cases to consider.
\begin{enumerate}
\item $S_{i+3}$ is a bad spine-subtree and $S_{i+2}$ is of type $\XXX_b$.\\
In that case, we get that $\tf(\ell)\le 2$ for every 2-leaf $\ell$ of $S_{i+2}$,
so that we can apply the standard modification of $\tf$ on $S_i$, 
$S_{i+1}$, $S_{i+2}$ and $S_{i+3}$, without decreasing the cost of $\tf$.
\item $S_{i+3}$ is a bad spine-subtree and $S_{i+2}$ is of type $\FFFA$.\\
In that case, $S_{i+4}$ must exist and must be of type $\XXX_c$,
since otherwise the sequence $S_iS_{i+1}S_{i+2}$ would be of type $\TTT_4$. 
Observe then that $S_{i+3}$ is somehow ``in the same situation as $S_i$''.

Let $L'=S'_1S'_2S'_3S'_4\dots S'_s$, $s\ge 5$, be the maximal subsequence of $L$, starting at $S_i$ 
(that is, $S'_1=S_i$), whose
type is a prefix of $(\FFFA.\XXX_c.\FFFA)^*$ (considered as a word),
and such that $S'_j$ is a bad spine-subtree if $j\equiv 1\pmod 3$.
We then have three cases, according to the value of $(s\mod 3)$.
\begin{enumerate}
\item If $s\equiv 1\pmod 3$, which means that $S_{i+s}$ is not of type $\XXX_c$, we get
that $S_i\dots S_{i+s-1}$ is a sequence of type $\TTT_4$, a contradiction.
\item If $s\equiv 0\pmod 3$, which means that $S_{i+s}$ is not a bad spine-subtree, we 
can apply the standard modification of $\tf$ on all the spine-subtrees 
$S'_j$ of $L'$ with $j\not\equiv 0\pmod 3$, without decreasing the cost of $\tf$.
\item If $s\equiv 2\pmod 3$, which means that $S_{i+s}$ is not of type $\FFFA$, 
we get that $S_{i+s}$ is thus of type $\XXX_b$.
If $S_{i+s+1}$ is not a bad spine-subtree, then we
can apply the standard modification of $\tf$ on all the spine-subtrees 
$S'_j$ of $L'$ with $j\not\equiv 0\pmod 3$, without decreasing the cost of $\tf$.
If $S_{i+s+1}$ is a bad spine-subtree, then we
can apply the standard modification of $\tf$ on $S_{i+s}$, $S_{i+s+1}$, and on all the spine-subtrees 
$S'_j$ of $L'$ with $j\not\equiv 0\pmod 3$, without decreasing the cost of $\tf$.
\end{enumerate}

\end{enumerate}

\item $S_{i+1}$ is of type $\XXX_b$ or $\XXX_c$, and $S_{i+2}$ is a bad spine-subtree.\\
%
%
In this case, we get that either $S_{i+3}$ is of type $\XXX_b$ or $\XXX_c$,
or $S_{i+3}$ is of type $\FFFA$ and $S_{i+4}$ is of type $\XXX_c$,
since otherwise $S_iS_{i+1}S_{i+2}$ would be a sequence of type $\TTT_4$.
Therefore, $S_{i+2}$ is ``in the same situation as $S_i$''.

Let $L'=S'_1S'_2S'_3\dots S'_s$, $s\ge 4$, be the maximal subsequence of $L$, starting at $S_i$ 
(that is, $S'_1=S_i$), whose
type is a prefix of $(\FFFA.(\XXX_b[\XXX_c))^*$ (considered as a word),
and such that $S'_j$ is a bad spine-subtree if $j\equiv 1\pmod 2$.
We then have two cases, according to the value of $(s\mod 2)$.
\begin{enumerate}
\item If $s\equiv 1\pmod 2$, which means that $S_{i+s}$ is not of type $\XXX_b$ nor $\XXX_c$,
we get that $S_{i+s}$ is of type $\FFFA$ and $S_{i+s+1}$ is of type $\XXX_c$,
since otherwise $L'$ would be a sequence of type $\TTT_4$.
In that case, since $S_{i+s-1}$ is a bad spine-subtree, we get that $\tf^*(S_{i+s+1}) = 2$,
so that we can apply the standard modification of $\tf$ on the spine-subtrees
$S_i,\dots,S_{i+s}$, without decreasing the cost of $\tf$.

\item If $s\equiv 0\pmod 2$, which means that $S_{i+s}$ is not a bad spine-subtree,
we consider two cases.

If $S_{i+s-1}$ is of type $\XXX_b$, then we can apply the standard modification 
of $\tf$ on the spine-subtrees $S_i,\dots,S_{i+s-1}$, without decreasing the cost of $\tf$.

If $S_{i+s-1}$ is of type $\XXX_c$,
we get that $S_{i+s}$ is of type either $\FFFA$ (but not a bad spine-subtree), 
$\XXX_b$ or $\YYYC$.
If $S_{i+s+1}$ is not a bad spine-subtree, then, again, we can apply the standard modification 
of $\tf$ on the spine-subtrees $S_i,\dots,S_{i+s-1}$, without decreasing the cost of $\tf$.
If $S_{i+s+1}$ is a bad spine-subtree and $S_{i+s}$ is not of type $\FFFA$,
then we necessarily have 
$\tf(\ell)\le 2$ for every 2-leaf $\ell$ of $S_{i+s}$,
so that
we can apply the standard modification of $\tf$ on the spine-subtrees 
$S_i,\dots,S_{i+s+1}$, without decreasing the cost of $\tf$.
If $S_{i+s+1}$ is a bad spine-subtree and $S_{i+s}$ is of type $\FFFA$,
then, by applying the standard modification of $\tf$ on $S_{i+s}$ and $S_{i+s+1}$,
we get a new subsequence $L'$, whose length has been increased by~1,
so that we now have $s\equiv 1\pmod 2$ and the previous case applies.

\end{enumerate}

\end{enumerate} 

This completes the proof.
\end{proof}

\subsection{Main result}
\label{sub:main-result}

We are now able to prove the main result of our paper.

\begin{theorem}
For every locally uniform $2$-lobster $L$ of length $k\ge 0$, 
$\beta_b(L)=\beta^*(L)$.
\label{th:main}
\end{theorem}

\begin{proof}
If $k=0$, the result follows from Lemma~\ref{lem:beta-star-zero}, observing that
the independent broadcast built in its proof reaches the upper bounds on the broadcast values
stated in Lemma~\ref{lem:not-more-than-beta-star}.
We can thus assume $k\ge 1$.
By Lemma~\ref{lem:beta-star-general}, we know that there exists an independent
broadcast $f$ on $L$ with $\cost(f)=\beta^*(L)$.
Let $f$ be the independent broadcast on $L$ constructed in the proof of Lemma~\ref{lem:beta-star-general}.
We claim that for every spine-subtree $S_i$ of $L$, $f^*(S_i)$ equals the upper
bound given in Lemmas \ref{lem:not-more-than-beta-star} or~\ref{lem:not-more-than-beta-star-total},
which will prove the theorem.
\begin{enumerate}
\item
If $S_i$ is of type $\GGG$ or $\XXX_a$, then $f^*(S_i)$ has been set to 
$\lambda_1(S_i)$ in Step~1, and never modified in the following steps.
\item
If $S_i$ is of type $\XXX_b$, then $f^*(S_i)$ has been set 
to $\lambda_2(S_i) + \lambda_2^*(S_i) + \alpha_2^*(S_i)$ in Steps 1 and~2.
Moreover, if $S_i$ belongs to some sequence of type $\TTT_4$, then $f^*(S_i)$ has been
decreased by~1 in Step~4.
\item
If $S_i$ is of type $\XXX_c$, then $f^*(S_i)$ has been set to $3$ in Steps 1 and~3.
Moreover, if $S_i$ belongs to some sequence of type $\TTT_4$, then $f^*(S_i)$ has been
decreased by~1 in Step~4.
\item
If $S_i$ is of type $\YYYC$, then $f^*(S_i)$ has been set to
$\lambda_2(S_i) + \lambda_2^*(S_i) + \alpha_2^*(S_i)$ in Steps 1 and~2,
and never modified in the following steps.
\item
Finally, if $S_i$ is of type $\FFFA$, then $f^*(S_i)$ has been set to 
$\lambda_1(S_i)$ in Step~1.
Moreover, if $S_i$ belongs to some sequence of type $\TTT_4$, then $f^*(S_i)$ has been
increased by~1 in Step~4.
\end{enumerate}
This concludes the proof.
\end{proof}

\section{Concluding remarks}
\label{sec:discussion}

In this paper, we have given an explicit formula
for the broadcast independence number of
a subclass of lobsters, called locally uniform 2-lobsters.
Moreover, it is easily seen that computing the value
$\beta^*(L)$ for a locally uniform 2-lobster
$L$ of length $k$
can be done in linear time (simply processing the spine-subtrees
$S_0,\dots,S_k$ in that order), which improves the result
of Bessy and Rautenbach~\cite{BR18a} for this particular subclass of trees.

A natural question, as a first step, would be to extend our result to
the whole subclass of locally uniform lobsters.
In fact, we were able to give an explicit formula for every such lobster
not containing any spine-subtree of type $\ZZZ$, that is,
having exactly one branch and three 2-leaves
(see~\cite{thesis}).
However, the proof is then quite involved and we thus decided to only consider in this paper
the restricted class of locally uniform 2-lobsters.
Determining when the optimal broadcast value of a spine-subtree of type $\ZZZ$
is $3$ or~$4$ appears to be not so easy.

The more general question of giving an explicit formula for the broadcast independence number of
the whole class of lobsters is certainly more challenging.


\end{document}